\documentclass[reqno,12pt]{amsart}

\usepackage{amsmath,amsthm,amssymb,amsfonts,epsfig,color,graphicx,enumerate,accents,subfigure,stackrel}
\usepackage{soul}
\usepackage[a4paper,margin=3truecm]{geometry}
\usepackage[T1]{fontenc}

\DeclareMathOperator{\spt}{spt}

\DeclareMathOperator{\dist}{dist}

\def\rhotild{\tilde{\rho}}

\def\ds{\displaystyle}
\def\eps{{\varepsilon}}

\def\R{\mathbb{R}}

\def\O{\Omega}

\def\C{\mathcal{C}}

\def\M{\mathcal{M}}

\def\rhotild{\tilde{\rho}}

\def\Cbar{C}
\newcommand{\be}{\begin{equation}}
\newcommand{\ee}{\end{equation}}
\newcommand{\bib}[4]{\bibitem{#1}{\sc#2: }{\it#3. }{#4.}}

\newcommand{\weak}{\stackrel{*}{\rightharpoonup}}
\newcommand{\norm}[1]{\left\lvert #1 \right\rvert}
\newcommand{\Norm}[1]{\left\lVert #1 \right\rVert}

\numberwithin{equation}{section}
\theoremstyle{plain}
\newtheorem{theo}{Theorem}[section]
\newtheorem{lemm}[theo]{Lemma}

\newtheorem{prop}[theo]{Proposition}

\theoremstyle{remark}
\newtheorem{rema}[theo]{\bf Remark}
\newtheorem{exam}[theo]{\bf Example}

\newenvironment{ack}{{\bf Acknowledgements. }}

\title[To write later]{Dissociation limits in Density Functional Theory}

\author{Guy Bouchitt\'e, Giuseppe Buttazzo, Thierry Champion, Luigi De Pascale}

\date{\today}

\begin{document}

\begin{abstract}In this paper we consider the {\it Density Functional Theory} (DFT) framework, where a functional of the form
$$F_\eps(\rho)=\eps T(\rho)+bC(\rho)-U(\rho)$$
has to be minimized in the class of non-negative measures $\rho$ which have a prescribed total mass $m$ (the total electronic charge). The parameter $\eps$ is small and the terms $T$, $C$, $U$ respectively represent the kinetic energy, the electronic repulsive correlation, the potential interaction term between electrons and nuclei. Several expressions for the above terms have been considered in the literature and our framework is general enough to include most of them.

It is known that in general, when the positive charge of the nuclei is small, the so-called {\it ionization phenomenon} may occur, consisting in the fact that the minimizers of $F_\eps$ can have a total mass lower than $m$; this physically means that some of the electrons may escape to infinity when the attraction of the nuclei is not strong enough.

Our main goal, continuing the research we started in \cite{bbcd18}, is to study the asymptotic behavior of the minimizers of $F_\eps$ as $\eps\to0$. We show that the $\Gamma$-limit functional is defined on sums of Dirac masses and has an explicit expression that depends on the terms $T$, $C$, $U$ that the model takes into account.

Some explicit examples illustrate how the electrons are distributed around the nuclei according to the model used.
\end{abstract}

\maketitle

\textbf{Keywords: } Density functional theory, multi-marginal optimal transport, duality and relaxation, Coulomb cost, quantization of minimizers

\textbf{2020 Mathematics Subject Classification:} 49J45, 49N15, 82M30, 35Q40

\section{Introduction}\label{sintro}

The study of atomic configurations, or more generally of configurations of molecules composed of several nuclei, is a key problem in quantum chemistry. A large number of mathematical models have been proposed as approximations of the Schr\"odinger energy, which seems numerically too difficult to handle in its original form. Here we consider the framework of {\it Density Functional Theory} (DFT) which consists in the identification of a functional which depends only on the density of electrons, and on its minimization among measures defined on the Euclidean space $\R^3$.

In its general form, for a general density $\rho$ of electrons, the total energy $F_\eps(\rho)$ to be minimized is the sum of three terms:
$$F_\eps(\rho)=\eps T(\rho)+bC(\rho)-U(\rho)$$
where $b$ is a fixed positive parameter, $\eps$ is a small positive parameter depending on the Planck constant $\hbar$, and $T$, $C$, $U$ respectively denote the kinetic energy, the electron correlation, and the potential energy. We do not consider the nucleus-nucleus interaction because we assume that the nuclei are fixed and this extra term would then simply be a constant in the minimization procedure. In addition, we assume that every nucleus is point-like, neglecting the repulsive terms which arise in nature among nuclear particles with positive charge. For more details on the physical interpretation of this model we refer to \cite{BBL}, \cite{bbcd18}, \cite{bdpgg12}, \cite{cfm14}, \cite{cfk17} and to references therein.

\medskip\noindent
$\bullet$\ For the {\it kinetic energy} we take an integral expression of the form
\be\label{kinetic}
T(\rho)=\int\rho^\alpha|\nabla\rho|^\beta\,dx\;,
\ee
with
$$\beta\ge0,\qquad3\alpha+4\beta=5,$$
and where we use the convention that integrals without domain specification are taken over the whole ambient space $\R^3$. Particular cases are:
\begin{itemize}
\item[-]the von Weizs\"acker energy, where $\beta=2$ and $\alpha=-1$;
\item[-]the Thomas-Fermi energy, where $\beta=0$ and $\alpha=5/3$.
\end{itemize}
Also multiples of terms as in \eqref{kinetic} can be considered, as well as their sums. The condition all these expressions verify is the scaling property
$$T(\rho_s)=s^2T(\rho)\qquad\text{where }\rho_s(x)=s^3\rho(sx).$$

\medskip\noindent
$\bullet$\ The {\it correlation term} $C$ represents the electron-electron repulsion. In the literature functionals of the form
\be \label{defC0D} C_0(\rho)=\frac34\int\rho^{4/3}dx\qquad\text{or}\qquad D(\rho)=\frac12\int\int\frac{1}{|x-y|}d\rho(x)d\rho(y)\ee
have been considered, as well as their combinations (Lieb-Oxford functionals \cite{LiebOxfo81})
$$C(\rho)=c_1\int\rho^{4/3}dx+c_2\int\int\frac{1}{|x-y|}d\rho(x)d\rho(y)\;.$$
Alternatively, the {\it Strongly Correlated Electrons} functional with $N$ electrons can be considered (see for instance \cite{bdpgg12,cfk13,gori2009density,gori2010density, seidl1999strong, seidl2007strictly,seidl1999strictly}); it is defined through a multi-marginal mass transport functional\footnote{For all the details on optimal transport theory we refer to the monographs \cite{santa} and \cite{vi03}.}
\be\label{defC}
C^N_{SCE}(\rho):=\inf\left\{\int_{\R^{3N}}\!\!\!c(x_1,\dots,x_N)\,dP(x_1,\dots,x_N)\, :\, \pi^\#_iP=\rho,\ 1\le i\le N\right\},
\ee
where $P$ is a nonegative measure on $\R^{3N}$, $\pi_i$ is the projection map from $\R^{3N}$ on its $i$-th factor $\R^3$, $\#$ denotes the push forward operator defined for a map $f$ and a measure $\mu$ by
$$f^\#\mu(E)=\mu\big(f^{-1}(E)\big),$$
and $c$ is the Coulomb correlation function
$$c(x_1,\dots,x_N)=\sum_{1\le i<j\le N}\frac{1}{|x_i-x_j|}\;.$$
As noticed in \cite{bbcd18,bbcd19}, the functional $C^N_{SCE}$ is not lower semicontinuous in the space of measures, which may produce a loss of mass at infinity, and a relaxation procedure is necessary to have a well posed optimization problem (see \cite{bbcd19}). Properties of $C^N_{SCE}$, which escapes from some foundational theorem in optimal transport \cite{kel84}, have been widely studied in the last ten years, see for example \cite{bcdp17, cdms18, cope19, dep15}.

When the number $N$ of electrons is not a priori prescribed, the so-called {\it Grand Canonical} functional is the natural extension of $C^N_{SCE}$. We recall its definition in Subsection \ref{canonical}. On the other hand, it is well known that the limit of $\frac1{N^2} C^N_{SCE}$ as $N\to\infty$ is the direct energy $D$ given in \eqref{defC0D} (see for instance \cite{SL, SP, Ssurvey} and \cite{bb22} for more general costs). 

Similarly to the kinetic energy term, all the models above satisfy a common scaling property, namely:
$$C(\rho_s)=s\,C(\rho)\qquad\text{where }\rho_s(x)=s^3\rho(sx) .$$

\medskip\noindent
$\bullet$\ The {\it potential term} $-U$ represents the attractive effect between electrons and nuclei. It is of the form
\be\label{U}
U(\rho)=\int V(x)\,d\rho(x)
\ee
being $V(x)$ the Coulomb potential
$$V(x)=\sum_{1\le k\le M}\frac{Z_k}{|x-X_k|}\,.$$
Here $M$ is the number of nuclei, $X_k$ is the location of the $k$-th nucleus, and $Z_k>0$ is its positive charge.
Notice that, in the case of a single nucleus at the origin, the potential energy $U$ enjoys the same scaling property as $C$ namely
\be\label{defU0}
U_0(\rho_s)=sU_0(\rho) \qquad\text{where $U_0(\rho):= \int \frac{\rho(x)}{|x|}\, dx$}.
\ee
Other kind of potentials $V(x)$ have been considered in \cite{bbcd19}. 

Our main goal is the study of the behavior of minimal or almost minimal configurations $\rho_\eps$ of the energy functional $F_\eps$ as $\eps\to0$. Due the singularity of the Coulomb potentials, it can be checked that the infima of $F_\eps$ go to $-\infty$ at most like $-1 / \eps$ when $\eps\to0$ for the correlation terms $C$ listed above (see Remark \ref{rema:infFeps}). It is then convenient to rescale the total energies $F_\eps$ and to consider the functionals
$$G_\eps(\rho)=\eps F_\eps(\rho),$$
which have the same minimizing sequences $\rho_\eps$. The characterization of the $\Gamma$-limit $G$ of the family $G_\eps$ allows then to establish the behavior of $\rho_\eps$ as $\eps\to0$.

The case $C(\rho)=C^N_{SCE}(\rho)$ was considered in \cite{bdp17}, \cite{bbcd18}, \cite{bbcd19} and \cite{lew17} when $N=2$ (two electrons), while the general case $N>2$ involves more difficult issues related to some expected sub-additivity property of the relaxed interaction energy. In the present paper we consider the other situations, specializing to the cases where
$$C(\rho)=C_0(\rho)\qquad\text{or}\qquad C(\rho)=D(\rho).$$
We prove that the $\Gamma$-limit energy $G$ is finite only on the set of discrete measures supported by the nuclei $X_1,\dots, X_M$ and given by
$$
G(\rho)=\sum_{1\le k\le M}g_b(Z_k,\alpha_k)\qquad\text{where }\alpha_k=\rho(\{X_k\}),
$$
where $g_b$ is a function suitably defined on $\R_+\times\R_+$, see \eqref{defgb} in Section \ref{s:main}. An important consequence of this formula is the so-called dissociation effect: as $\eps\to0$ a system with $M$ nuclei behaves like $M$ independent systems made by only one nucleus with charge $Z_k$ and $\alpha_k$ electrons\footnote{here $\alpha_k$, not necessarily integer, denotes the expectation of the number of electrons attached to the $k$-th nucleus.}, $k=1,\dots,M$. The asymptotic study then reduces to the finite dimensional minimum problem
$$\min\bigg\{\sum_{1\le k\le M}g_b(Z_k,\alpha_k)\ :\ \sum_{1\le k\le M}\alpha_k\le m\bigg\},$$
where $m$ is the total electronic charge. In this model the charges $Z_k$ are prescribed and the given positions $X_k$ of the nuclei do not play any role since no repulsive effect between them is considered.

The determination of the function $g_b(Z,\alpha)$ is fundamental; in order to compute it, we introduce an auxiliary variational problem where the upscaled total energy associated to a single nucleus of charge $Z$ is minimized over all non negative measures whose total mass does not exceed a parameter $t$. More precisely we will consider the nonincreasing function $L:\R_+\to(-\infty,0]$ defined by
$$L(t)=\inf\left\{T(\rho)+C(\rho)-\int\frac{\rho(x)}{|x|}\,dx\ :\ \rho\ge0,\ \int\rho(x)\,dx\le t\right\}.$$
An important question is whether the function $L(t)$ is constant after a suitable threshold $t_*$. In fact, this case encodes the so-called {\it ionization} phenomenon, in which some of the electrons surrounding a nucleus located at the origin escape to infinity. As we will see, the existence of such a $t_*$ depends on the choice of the model for the kinetic energy $T$ and for the electronic interaction $C$. This ionization effect has been widely studied and we refer to \cite{fnvdb18}, \cite{fnvdb18b}, \cite{solo91}, \cite{solo03} and references therein for all the details and the corresponding results and conjectures.


The plan of the paper is as follows.
\begin{itemize}
\item[-]In Section \ref{s:hypo}, we give in details the precise assumptions on the energy terms $T(\rho)$ and $C(\rho)$ and outline their main properties.
\item[-]In Section \ref{s:main} we state the main convergence result, whose proof is postponed to Section \ref{s:proofG}.
\item[-]In Section \ref{s:mono} we focus on the case of only one nucleus and study the functions $L(t)$ and $g_b(Z,\alpha)$, and the relation between them. The existence of a threshold $t_*$ will be then discussed in the different models.
\item[-]Finally, Section \ref{s:examplesnew} is dedicated to the presentation of two examples.
\end{itemize}

\section{Main assumptions and properties of energy terms}\label{s:hypo}

All along the paper we will denote by $\M=\M(\R^3)$ the Banach space of signed Radon measures over $\R^3$ endowed with the total variation norm; the weak* convergence $\rho_\eps\weak\rho$ in $\M$ is understood with respect to the duality with the space $C_0(\R^3)$ of continuous functions vanishing at infinity. The subset consisting of non-negative elements of $\M$ will be denoted by $\M_+$.

For every $\eps>0$, we define $F_\eps:\M\to\R\cup\{+\infty\}$ by:
\be\label{defFeps}
F_\eps(\rho)=\begin{cases}
\eps T(\rho)+bC(\rho)-U(\rho)&\text{if $\rho\in \M_+$ and $U(\rho)<+\infty$}\\
+\infty&\text{otherwise}.
\end{cases}\ee
We now list the main assumptions we will use on the energy terms $T$, $C$ and $U$ that form this energy functional.

\subsection{General hypotheses on the kinetic energy} \label{s:hypoT}
We are interested in kinetic energies of integral form
\be\label{Tgeneral}
T(\rho) = \int f(\rho(x),\nabla \rho(x)) \,dx\,.
\ee
More precisely, we shall mainly focus on the examples of the forthcoming Section \ref{s:mono} related to the particular case (see Subsection \ref{s:kineticW} below)
$$f(r,\xi)=r^\alpha|\xi|^\beta\qquad\text{for }(r,\xi)\in\R_+\times\R^3.$$
However, we state here our hypotheses for the more general case of functionals of the form \eqref{Tgeneral} above.
We hereafter assume that the integrand $f$ ensures the following fundamental properties for $T$.
\begin{itemize}
\item \emph{Convexity and lower semicontinuity}: $T:\M_+\to[0,+\infty]$ is convex and weak* lower semi-continuous (in duality with $C_0(\R^3)$).
\item \emph{Finiteness}: $T(\rho)<+\infty$ for every $\rho\in C^\infty_c(\R^3)$.
\item \emph{Scaling}: for all $\rho\in\M_+$ and $s>0$ it holds
\newcounter{TTT} \setcounter{TTT}{-1}
\be\tag{T\theTTT}\stepcounter{TTT}\label{h:Thomog}
T(\rho_{\#s})=s^{2}\,T(\rho)\qquad\text{with }\rho_{\#s}:=(h_s)_\#\,\rho,\text{ where }h_s(x)=x / s
\ee
\end{itemize}
Note that if $\rho$ is absolutely continuous then $\rho_{\# s} : x \mapsto s^{3} \,\rho(s x)$.

From the general form \eqref{Tgeneral} we easily get that $T$ is translation invariant, meaning
$$T(\rho(\cdot+X))=T(\rho)\qquad\text{for every }X\in\R^3,$$
and that it is also additive in the sense
\be \tag{T\theTTT}\stepcounter{TTT}\label{h:Tadd}
\forall\mu,\nu\in\M_+,\quad\spt(\mu)\cap\spt(\nu)=\emptyset\quad\Longrightarrow\quad T(\mu+\nu)=T(\mu)+T(\nu)\,.
\ee
Moreover, these functionals usually have the following \emph{truncation} property (under mild assumptions on $f$), that we shall also subsequently assume: there exists a sequence of smooth cut-off functions $\theta_n$ with $0\le\theta_n\le1$, $\theta_n=1$ on $B(0,n)$ and $\theta_n=0$ outside $B(0,n+1)$ such that
\be\tag{T\theTTT}\stepcounter{TTT}\label{h:Ttrunc}
\forall\rho\in\M_+,\qquad T(\theta_n\,\rho)\to T(\rho).
\ee
Finally we require two less usual properties for $T$. First, our analysis relies on the fact that the kinetic energy $T$ locally controls the potential energy $U_0$ (see \eqref{defU0}) in the following way:
\be\tag{T\theTTT}\stepcounter{TTT}\label{h:TU}
U_0\left(1_\Omega\,\rho\right) =\int_\O\frac{1}{|x|}\,\rho(dx)\le K_U \,(T(\rho)+1)^q\,\left(\int_\O \,d\rho\right)^p,
\ee
for every $\rho\in \M_+$ and $\Omega$ open subset of $\R^3$, for some constants $q \in \,]0,1[\,$, $p \in \,]0,1]$ and $K_U >0$ which are independent of the open set $\Omega$. Second, we also need some \emph{localization property} property for $T$, in the sense that for every $\delta >0$,
there exists a smooth cut-off function $\theta_\delta$ around $\{X_1,\ldots,X_M\}$ (\emph{i.e.} $0\le\theta_\delta\le1$, $\theta_\delta$ equals $1$ on $\cup_i B(X_i,\delta)$ and $0$ outside $\cup_i B(X_i,2\delta)$) such that
\be\tag{T\theTTT}\stepcounter{TTT}\label{h:Tlocal}
T(\theta_\delta\,\rho)\le T(\rho) + K_T (T(\rho) + 1)^r \omega(\delta)
\ee
for every $\rho\in L^1_+(\R^3)$, for some constants $K_T$ and $r \in [0,1]\,$, with $\omega(\delta) \to 0$ as $\delta \to 0$.


\subsection{Validity of the hypotheses for a class of kinetic energies}
\label{s:kineticW}
We consider kinetic energies of the form \eqref{kinetic}, that is
\be\label{h:kinetic}
T(\rho)=\int\rho^\alpha \norm{\nabla \rho}^\beta dx.
\ee
To ensure weak* lower semi-continuity of $T$ on $\M_+$ it is necessary that $\beta \geq 0$.
The scaling property \eqref{h:Thomog} reduces to $3 \alpha + 4\beta =5$, and the convexity of $T$ is equivalent to $\alpha\notin\,]0,1[$, $\beta\notin\,]0,1[$ and $\alpha\beta(1-\alpha-\beta)\ge0$. As a consequence, we shall hereafter assume
\be\label{h:alphabeta}
\alpha=\frac{5}{3} \;\text{ and }\; \beta=0 \qquad \text{or} \qquad 
\alpha=\frac{5-4\,\beta}{3} \;\text{ and }\; \beta\in[5/4,2]\,.
\ee
The first case is known as Thomas-Fermi model, and we use the notation
$$T_0(\rho)=\frac{3}{5}\int\rho^{5/3}dx.$$
The second case is of the von Weizs\"acker form, and we use the notation
$$W_{\alpha,\beta}(\rho)=\int\rho^\alpha|\nabla\rho|^\beta dx\qquad \mbox{with} \quad \alpha=\frac{5-4\beta}{3},\ \beta\in[5/4,2].$$
For this case we derive the following estimate for future use.

\begin{lemm}\label{lemm:Wsobol}
Assume \eqref{h:alphabeta} with $\beta\in[5/4,2]$, then there exist a constant $\kappa$ such that
$$\forall \rho \in L^1_+(\R^3), \qquad \Norm{\rho}_{L^{(5-\beta)/(3-\beta)}(\R^3)} \,\le\kappa W_{\alpha,\beta}(\rho)^{3/(5-\beta)}\,.$$
Note that $(5-\beta)/(3-\beta)\ge15/7$ for $\beta\in[5/4,2]$.
\end{lemm}

\begin{proof}
Under \eqref{h:alphabeta} with $\beta\in[5/4,2]$, it holds $\frac{\alpha+\beta}{\beta} \in \,]\frac{1}{2},1]$ and we infer that $u=\rho^{\frac{\alpha+\beta}{\beta}}$ is in $L^1_{loc}$. The classical Sobolev embedding inequality yields
$$\|u\|_{L^{3\beta/(3-\beta)}(\R^3)}\le\kappa_0\|\nabla u\|_{L^\beta(\R^3)}$$
for some constant $\kappa_0$. Then the identity
$$W_{\alpha,\beta}(\rho)=\left(\frac{\beta}{\alpha+\beta}\right)^\beta \int \norm{\nabla(\rho^{\frac{\alpha+\beta}{\beta}})}^\beta dx\, = \left(\frac{\beta}{\alpha+\beta}\right)^\beta \Norm{\nabla u}_{L^\beta(\R^3)}^\beta$$
implies the claim.
\end{proof}




The property \eqref{h:Ttrunc} is quite classical in this setting, while property \eqref{h:TU} is a direct consequence of the following result, which was stated in \cite[Lemma 3.6]{bbcd18} for the special case $\alpha=-1, \beta=2$. The proof below follows the same approach.

\begin{lemm}\label{lemm:TU}
Let \eqref{h:kinetic} and \eqref{h:alphabeta} hold, then there exist a constant $\kappa$ such that for every domain $\Omega$ (bounded or not) and every $\rho \in \M_+ \cap L^1(\R^3)$ we have
$$\int_\O\frac{1}{|x|}\rho(x)\,dx\le
\begin{cases}
\kappa\,T_0(\rho)^{1/2}\,\left(\int_\O\rho\,dx\right)^{1/6}&\text{if }\beta=0\\
\kappa\,\big(W_{\alpha,\beta}(\rho)\big)^{1/2}\Big(\int_\O\rho\,dx\Big)^{(\beta+1)/6}&\text{if }\beta\in[5/4,2].
\end{cases}$$
\end{lemm}

\begin{proof}
For the Thomas-Fermi case $\alpha=5/3$ and $\beta=0$ for all $\delta>0$ we have
\begin{align*}
\int_\O\frac{1}{|x|}\rho(x)\,dx
&\le\int_{\O\cap B(0,\delta)}\frac{1}{|x|}\rho(x)\,dx + \frac{1}{\delta}\int_{\O\setminus B(0,\delta)}\rho(x)\,dx\\
&\le\|\rho\|_{L^{5/3}}\|1/|x|\|_{L^{5/2}(B(0,\delta))} + \frac{1}{\delta}\int_\O\rho(x)\,dx\\
&=\kappa \ \delta^{1/5}\,T_0(\rho)^{3/5} + \frac{1}{\delta}\int_\O\rho(x)\,dx.
\end{align*}
Optimizing with respect to $\delta>0$ yields the result.

For the $W_{\alpha,\beta}$ case, we have, similarly,
\begin{align*}
\int_\O\frac{1}{|x|}\rho(x)\,dx
&\le\int_{\O\cap B(0,\delta)}\frac{1}{|x|}\rho(x)\,dx + \frac{1}{\delta}\int_{\O\setminus B(0,\delta)}\rho(x)\,dx\\
&\le\|\rho\|_{L^{(5-\beta)/(3-\beta)}}\|1/|x|\|_{L^{(5-\beta)/2}(B(0,\delta))} + \frac{1}{\delta}\int_\O\rho(x)\,dx\\
&=(4 \pi)^{2/(5-\beta)} \|\rho\|_{L^{(5-\beta)/(3-\beta)}}\delta^{(\beta+1)/(5-\beta)} + \frac{1}{\delta}\int_\O\rho(x)\,dx,
\end{align*}
From Lemma \ref{lemm:Wsobol} we obtain the existence of a constant $\kappa$ such that 
$$\int_\O\frac{1}{|x|}\rho(x)\,dx \;\le\; \kappa \, \left(W_{\alpha,\beta}(\rho)\right)^{3/(5-\beta)}\delta^{(\beta+1)/(5-\beta)} + \frac{1}{\delta}\int_\O\rho(x)\,dx.$$
Optimizing with respect to $\delta>0$ yields the result.
\end{proof}



We finally turn to the localization property \eqref{h:Tlocal}, which follows from Lemma \ref{lemm:Tlocal} below.

\begin{lemm}\label{lemm:Tlocal}
Let \eqref{h:kinetic} and \eqref{h:alphabeta} hold, then for all $\delta>0$ there exist a smooth cut-off function $\theta_\delta$ at the points $X_1,\ldots,X_M$ (\emph{i.e.} $0\le\theta_\delta\le1$, $\theta_\delta$ equals $1$ on $\cup_i B(X_i,\delta)$ and $0$ outside $\cup_i B(X_i,2\delta)$) and a constant $K$ such that
$$\forall\rho\in\M_+\cap L^1(\R^3),\qquad T(\theta_\delta\,\rho)\le T(\rho)+K \,T(\rho)\,\delta^\beta\,.$$
\end{lemm}

\begin{proof}
First note that in the case $\alpha=5/3$ and $\beta=0$ we obviously have $T(\theta_\delta\rho)\le T(\rho)$ for any smooth cut-off function $\theta_\delta$.

We now turn to the case $\beta\in\left[\tfrac{5}{4},2\right]$. For any smooth cut-off function $\theta$, we compute
\begin{align*}
T(\theta\rho)&=\int \theta^\alpha\rho^\alpha|\rho\nabla\theta+\theta\nabla\rho|^\beta\,dx\\
&\le\int\theta^\alpha\rho^\alpha\left((1-\theta)\norm{\frac{\nabla \theta}{1-\theta}\rho}^\beta + \theta |\nabla\rho|^\beta\right)dx\\
&\le\int\theta^{\alpha}(1-\theta)^{1-\beta}|\nabla\theta|^\beta\rho^{\alpha+\beta}\,dx+T(\rho)
\end{align*}
were we used the convexity of the function $t\mapsto t^\beta$. We choose $\theta$ so that $\theta^{\alpha} (1-\theta)^{1-\beta} \norm{\nabla \theta}^\beta$ is bounded by some constant $K_0$ (see Remark \ref{rmk:deftheta} below). Using the fact that $\theta$ is supported on $\cup_i B(X_i,\delta)$, we now apply the H\"older inequality and Lemma \ref{lemm:Wsobol} to get
\begin{align*}
\int \theta^{\alpha} (1-\theta)^{1-\beta} \norm{\nabla \theta}^\beta \rho^{\alpha+\beta}\,dx
&\le K_0\,\left(\int\rho^q dx\right)^{(3-\beta)/3}\left(\int_{\cup_i B(X_i,2\delta)}1\,dx\right)^{\beta/3}\\
&\le K\,W_{\alpha,\beta}(\rho)\,\delta^\beta\,.
\end{align*}
for some constant $K$, which concludes the proof.
\end{proof}

\begin{rema}\label{rmk:deftheta}
As in \cite[Example 4.2]{bbcd18}, we may construct an explicit cut-off $\theta_\delta$ by considering the real functions 
$$f(t)=\begin{cases}
e^{-1/t}&\mbox{if }t>0,\\
0&\mbox{otherwise.}
\end{cases}
\qquad\text{and}\qquad g(t)=\frac{f(t)}{f(t)+f(1-t)}\,.
$$
Then $\theta_\delta(x):=g\left(\norm{x}/\delta\right)$ is such that $\theta_\delta^p \, (1-\theta)^q\, \norm{\nabla \theta_\delta}^r$ is bounded whenever $r>0$. Since this choice does not depend on the parameters $\alpha,\beta$, then any positive linear combination of kinetic energies of the form \eqref{h:kinetic} satisfies the hypotheses of Subsection \ref{s:hypoT}.
\end{rema}

\subsection{General hypotheses on the electron-electron interaction term}
The relevant properties of this part of the energy for our study are the following.

\begin{enumerate}
\item {\it Convexity and lower semicontinuity}: $C:\M_+\to[0,+\infty]$ is convex and weak* lower-semicontinuous.
\item {\it Finiteness}: $C(\rho) < +\infty$ for any $\rho \in C^\infty_c(\R^3;\R_+)$.
\item {\it Translation invariance}: $C$ is translation invariant, meaning
$$C(\rho(\cdot+X))=C(\rho)\qquad\text{for every }X\in\R^3,$$
\item {\it Scaling}: for all $\rho\in\M_+$ and $s>0$ it holds
\newcounter{CCC} \setcounter{CCC}{-1}
\be\tag{C\theCCC}\stepcounter{CCC} \label{h:Chomog}
\Cbar(\rho_{\#s})=s\,\Cbar(\rho)
\qquad\text{with }\rho_{\#s}=(h_s)_\#\,\rho,\text{ where }h_s(x)=x / s.
\ee
\item {\it Monotonicity}: whenever $\rho_1\le\rho_2$ we have
\be\tag{C\theCCC}\stepcounter{CCC}\label{h:Cmono}
C(\rho_1)\le C(\rho_2).
\ee
\item {\it Superadditivity}: whenever $\rho_1$ and $\rho_2$ have supports at distance $R>0$,
\be\tag{C\theCCC}\stepcounter{CCC}\label{h:Csuperadd}
\Cbar(\rho_1+\rho_2)\ge\Cbar(\rho_1)+\Cbar(\rho_2)-K_C(R)\rho_1(\R^3)\rho_2(\R^3)
\ee
for some constant $K_C(R)$ depending only on $R$.
\item {\it Weak subadditivity}: whenever $\rho_1$ and $\rho_2$ have supports at distance $R>0$,
\be\tag{C\theCCC}\stepcounter{CCC} \label{h:Csubadd}
\Cbar(\rho_1+\rho_2)\le\Cbar(\rho_1)+\Cbar(\rho_2)+K_C(R)\rho_1(\R^3)\rho_2(\R^3).
\ee
with $K_C(R)\to0$ as $R\to+\infty$.
\item {\it Domination by $U_0$}: for any $b\ge0$ and $U_0$ given in \eqref{defU0}, it holds
\be\tag{C\theCCC}\stepcounter{CCC}\label{h:CU0}
\inf\left\{b\,C(\rho)-U_0(\rho)\ :\ \rho\in C^\infty_c(\R^3)\right\}<0.
\ee
\end{enumerate} 

\begin{rema}\label{rema:infFeps}
We first note that the homogeneity property \eqref{h:Chomog} of $C$ is also shared by $U_0$, so that \eqref{h:CU0} holds whenever the infimum is in fact $-\infty$.
We proceed as in \cite{bbcd18} to show that $\inf\{F_\eps\}$ goes to $-\infty$ at most like $-1/\eps$. By the translation invariance of $T$ and $C$, we infer that the functional $F_\eps$ given in \eqref{defFeps} satisfies
$$\inf\{F_\eps\}\le\inf\left\{\eps T(\rho)+bC(\rho)-Z_1 U_0(\rho)\ :\ \rho\in C^\infty_c(\R^3)\right\}.$$
Thanks to \eqref{h:CU0} there exists $\rho^*\in C^\infty_c(\R^3)$ such that $b\,C(\rho^*)-Z_1 U_0(\rho^*)<0$. Then from the respective homogeneity properties \eqref{h:Thomog} and \eqref{h:Chomog}
of $T$ and $C$ it follows that
\[
\forall s>0, \qquad \inf\{F_\eps\} \leq s^2 \,\eps T(\rho^*) + s\,b\,C(\rho^*)-s\, Z_1 U_0(\rho^*)\,.
\]
Optimizing with respect to $s>0$ yields that
\[
\inf\{ F_\eps\} \leq -\frac{1}{\eps} \,\frac{(b\,C(\rho^*)-\, Z_1 U_0(\rho^*))^2}{4\,T(\rho^*)}
\]
which proves the claim.
\end{rema}

\medskip
We now discuss several examples.

\subsubsection{The case $C_0$} 
The functional
$$C_0(\rho):=\frac{3}{4}\int\rho^{4/3}dx$$
obviously satisfies the properties above.

\subsubsection{The direct energy $D$}
By standard considerations it is possible to see that the functional (considered in \cite{lieb83})
$$D(\rho):=\frac12\int\!\!\!\int\frac{1}{|x-y|}\,d\rho(x)\,d\rho(y)$$
is convex and lower semi-continuous with respect to the weak* convergence on $\M_+$. Properties \eqref{h:Chomog} and \eqref{h:Cmono} are also satisfied.
Property \eqref{h:CU0} follows from the fact that
$$b\,D(t\,\rho)-U_0(t\,\rho)=b\,t^2\,D(\rho)-t U_0(\rho)<0$$
for any $t>0$ sufficiently small whenever $\rho \in C_c^\infty(\R^3;\R_+)$ is not equal to $0$.

We now turn to the additivity properties.

\begin{lemm}\label{lemm:Daddit} Let $\rho_1,\ \rho_2\in \M_+$ be such that $\displaystyle \delta :=\dist(\spt(\rho_1),\spt(\rho_2)) >0$, then 
$$D(\rho_1)+D(\rho_2)\le D(\rho_1+\rho_2)\leq D(\rho_1)+D(\rho_2)+\frac{2}{\delta}\rho_1(\R^3) \rho_2(\R^3).$$
\end{lemm}

\begin{proof}
It is enough to write 
$$D(\rho_1+\rho_2)= D(\rho_1)+D(\rho_2)+\int\!\!\!\int\frac{1}{|x-y|}\,d\rho_1(x)\,d\rho_2(y)$$
and the last term is non-negative and smaller than $2\rho_1(\R^3)\rho_2(\R^3)/\delta$.
\end{proof}

We shall also need the following estimates which relates the energy $D$ and the potential $U_0$. An estimate in the spirit of the first one was obtained in the non radial case in \cite[Lemma 2]{BBL}.

\begin{lemm}\label{lemm:DU}
For all $R>0$ it holds
$$\int_{|x|\ge R}\frac{\rho(x)}{|x|}\,dx\le\sqrt{\frac{2D(\rho)}{R}}.$$
Moreover, one has
\[
D(\rho)\ge\frac{1}{8\pi}\int_0^\infty\frac{\eta(r)^2}{r^2}\,dr \qquad \mbox{where } \ \eta(r)=\rho(B(0,r)).
\]
with equality whenever $\rho$ is a radial function.
\end{lemm}

\begin{proof}
For $R>0$ we compute
\[\begin{split}
D(\rho)&\ge\int_{|x|\ge R}\int_{|y|\ge R}\frac{\rho(x)\rho(y)}{|x-y|}\,dx\,dy\\
&\ge\int_{|x|\ge R}\int_{|y|\ge R}\frac{\rho(x)\rho(y)}{|x|+|y|}\,dx\,dy\\
&=\int_{|x|\ge R}\frac{\rho(x)}{|x|}\int_{|y|\ge R}\frac{\rho(y)}{|y|}\frac{|x||y|}{|x|+|y|}\,dx\,dy\\
&\ge\Big(\int_{|x|\ge R}\frac{\rho(x)}{|x|}\,dx\Big)^2\frac{R}{2},
\end{split}\]
which concludes the proof of the first inequality.

For the second inequality, we first note that since $D$ is convex we have \[
D(\rho) \geq D( \norm{\cdot}_\# \rho)
\]
so that it is sufficient to prove the equality case for a radial element $\rho \in \M_+$.
By approximation we can reduce to the case where $\rho$ is smooth with compact support, so that in particular $\eta$ is smooth and bounded. Then the radial solution $U=U(r)$ generated by the charge $\rho$, solution of $-\Delta U = \rho$ in $\R^3$, satisfies
$$\forall r\ge0,\qquad -(r^2 U')'=r^2\rho(r)\quad\mbox{and}\quad U(+\infty)=0.$$
Integrating on $(0,R)$ and then on $(R,+\infty)$, we get respectively
\[
- R^2 U'(R) = \frac1{4\pi}\eta(R)
\quad\mbox{ and }\quad U(R) = \frac 1{4\pi}\int_R^\infty \frac{\eta(s)}{s^2} ds\,.
\]
From $\eta'(r)= 4\pi r^2\rho(r)$ we infer
\begin{align*}
D(\rho)&=2\pi\int_0^\infty U(r)\rho(r)\,r^2\,dr=\frac12\int_0^\infty U(r)\eta'(r)\,dr\\
&=\frac12\left(\left[\eta(r)U(r)\right]_0^\infty-\int_0^\infty\eta(r)U'(r)\,dr\right)=\frac1{8\pi}\int_0^\infty\frac{\eta(r)^2}{r^2}\,dr\,,
\end{align*}
which concludes the proof.
\end{proof}

\subsubsection{The Grand Canonical $C_{GC}$}\label{canonical}
The Grand Canonical interaction cost $C_{GC}$ associated to the Coulomb cost may is given for $\rho \in \M_+$ by
$$C_{GC}(\rho)=\inf\left\{ \sum_{N\ge1} C^N_{SCE}(\rho_N) \,:\, \rho_N \in \M_+ , \;\sum_{N \ge1}\rho_N\le1, \;
\sum_{N\ge1}N\rho_N=\rho \right\},$$
where for all $N\ge2$ the functional $C_{SCE}^N$ is defined in \eqref{defC}, while for $N=1$ we define $C_{SCE}^1=0$. Note than whenever $\rho(\R^3)\le1$ we have $C_{GC}(\rho)=0$.

The fact that the cost $C_{GC}$ is convex and weak* lower-semicontinuous on $\M_+$ follows from Theorems 2.1 and 2.2 in \cite{dmlene22}. The fact that $C_{GC}(\rho)<+\infty$ for $\rho\in C^\infty_c(\R^3)$ is immediate, and the scaling property \eqref{h:Chomog} follows by standard changes of variables. Property \eqref{h:CU0} follows from the fact that $G_{GC}(\rho)=0$ whenever $\int d\rho\le1$. The last properties to verify are stated in the following lemma.

\begin{lemm}
The cost $C_{GC}$ satisfy properties \eqref{h:Cmono}, \eqref{h:Csuperadd} and \eqref{h:Csubadd}.
\end{lemm}

\begin{proof}
Properties \eqref{h:Cmono} and \eqref{h:Csuperadd} follow from the duality Theorem 4.2 in \cite{dmlene22}. Indeed, according to \cite{dmlene22}, the Fenchel conjugate functional of $C_{GC}$ in the duality $(\M,C_0)$ is given by
$$C^*_{GC}(\phi)=\sup\left\{\sum_{i=1}^N\phi(x_i)-c_N(x_1,\dots,x_N)\ :\ N\ge1,\ x_i\in\R^d,\ i=1,\dots,N\right\}.$$
We observe that we may increase the sum above by sending to infinity all the points $x_i$ such that $\phi(x_i)<0$; this gives
$$C^*_{GC}(\phi)=C^*_{GC}(\phi_+).$$
Thus the functional $C_{GC}(\rho)$ can be recovered as
\be\label{gc*}
C_{GC}(\rho)=\sup\left\{\int\phi\,d\rho-C^*_{GC}(\phi)\ :\ \phi\in C_0,\ \phi\ge0\right\}.
\ee
This gives easily the monotonicity property (C1). On the other hand, the functional $C^*_{GC}$ given above is subadditive, which implies that for any $\rho_1,\rho_2\in\M_+$ and any $\phi_1,\phi_2\in C_0(\R^d;\R_+)$ it holds
\[\begin{split}
C_{GC}(\rho_1+\rho_2)&\ge\int(\phi_1+\phi_2)\,d(\rho_1+\rho_2)-C^*_{GC}(\phi_1+\phi_2)\\
&\ge\int\phi_1\,d\rho_1+\int\phi_2\,d\rho_2-C^*_{GC}(\phi_1)-C^*_{GC}(\phi_2),
\end{split}\]
and the superadditivity property (C2) (with $K_C(R)=0$ and without assuming disjoint supports of $\rho_1$ and $\rho_2$) now follows by \eqref{gc*} taking the supremum over all pairs $\phi_1,\phi_2$.

Finally, property \eqref{h:Csubadd} follows from the full subadditivity of the functional $C_{GC}(\rho)+D(\rho)$ (see equation (3.2) in \cite{LLS}) and the fact that $D(\rho)$ satisfies \eqref{h:Csubadd}.
\end{proof}

\subsubsection{The Strictly Correlated Electrons $C^N_{SCE}$}
As pointed out in \cite{bbcd18,bbcd19}, the functional $C^N_{SCE}$ is not lower-semicontinuous on $\M_+$: in that papers, this difficulty is bypassed by considering its weak* semi-continuous envelope $\overline{C}^N_{SCE}$ over the set of subprobabilities. Then it is shown in \cite{bbcd18} that this relaxed functional $\overline{C}^N_{SCE}$ satisfies all the required properties in the special case $N=2$. It also follows from \cite{bbcd19} that, for $N\ge3$, all the properties listed above hold, with the exception of the weak subadditivity \eqref{h:Csubadd}, which is still an open question.

\vskip 2em

\section{Statement of the $\Gamma$-convergence result}\label{s:main}

The main convergence result of this paper is stated in the framework of $\Gamma$-convergence theory, for which we refer to \cite{DM}. In the following $\Gamma$-convergence result, the ambient space is $\M_+(\R^3)$ endowed with the weak* topology associated to the duality with $C_0(\R^3)$. The limit functional $G$ is given in \eqref{defG} by
\be\label{defG}
G:\rho\mapsto G(\rho):=\sum_{i=1}^N g_b(Z_i,\rho(\{X_i\}))
\ee
where the functions $g_b$ are defined on $\R_+^2$ by
\be\label{defgb}
g_b(Z,\alpha)=\inf\left\{T(\mu)+b\,C(\mu)-Z\int\frac{1}{|x|}\,\mu(dx)\ :\ \mu \in \M_+,\ \int d\mu\le\alpha\right\}.
\ee

\begin{theo}\label{th:mainG}
Under the assumptions \eqref{h:Thomog}-\eqref{h:Tlocal} and \eqref{h:Chomog}-\eqref{h:Csubadd}, we have that the sequence of functionals over $\M_+$ given by
$$G_\eps:\rho\mapsto G_\eps(\rho)=\eps F_\eps(\rho)=\eps^2 T(\rho)+\eps C(\rho)-\eps U(\rho)$$
$\Gamma$-converges to the functional $G$ given in \eqref{defG} that is
\begin{enumerate}[(i)]
\item$\displaystyle G(\rho)=\Gamma-\liminf G_\eps(\rho):=\inf\big\{\liminf_{\eps\to0}G_\eps(\rho_\eps)\ :\ \rho_\eps\weak\rho\big\}$;
\item$\displaystyle G(\rho)=\Gamma-\limsup G_\eps(\rho):=\inf\big\{\limsup_{\eps\to0}G_\eps(\rho_\eps)\ :\ \rho_\eps\weak\rho\big\}$.
\end{enumerate}
In addition, for every $\rho\in\M_+$ , there exists a recovering sequence $(\rho_\eps)$ weakly* converging to $\rho$ such that $\int d\rho_\eps=\int d\rho$ and 
$$\lim_{\eps\to0}G_\eps(\rho_\eps)=G(\rho)=\sum_{i=1}^N g_b(Z_i,\rho(\{X_i\}))\,.$$
\end{theo}

\begin{rema}
A similar result was proved in \cite{bbcd18} in the particular case where $T(\rho)=W_{-1,2}(\rho)$ and $C(\rho)=C_{SCE}^2(\rho)$, that is
$$F_\eps(\rho)=\eps W_{-1,2}(\rho)+b\,C_{SCE}^2(\rho)-U(\rho).$$
\end{rema}

\medskip
As a consequence of Theorem \ref{th:mainG} and of general properties of the $\Gamma-$convergence, the minimizers $\rho_\eps$ of
$$\min\left\{F_\eps(\rho)\ :\ \int d\rho\le m\right\}$$
weakly* converge to a minimizer $\rho$ of the problem
\be\label{limpb}
\min\left\{G(\rho)\ :\ \int d\rho\le m\right\}.
\ee
Therefore, the asymptotic behavior of $\rho_\eps$ as $\eps\to0$ can be obtained by studying the minimization problem
$$\min\Big\{g_b(Z_1,\alpha_1)+\dots+g_b(Z_m,\alpha_m)\ :\ \alpha_1+\dots+\alpha_M\le m\Big\}.$$
In particular we see that the minimizers $\rho$ of \eqref{limpb} of minimal total mass are sums of Dirac masses concentrated on the set of nuclei $\{X_1,\dots,X_M\}$.

\medskip
Some general properties of the function $g_b$ are listed below.

\begin{prop}
The following properties of the function $g_b$ hold.
\begin{itemize}
\item[(i)]The function $g_b(Z,\cdot)$ is convex, continuous, non-increasing on $\R_+$, with $g_b(Z,0)=0$, and such that $g_b(Z,\alpha)<0$ when $\alpha>0$.
\item[(ii)]The function $g_b(\cdot,\alpha)$ is concave, continuous, non-increasing, with $g_b(0,\alpha)=0$, and such that $\lim_{Z\to+\infty}g(Z,\alpha)=-\infty$ for every $\alpha>0$.
\end{itemize}
\end{prop}


\begin{proof}
We start to prove (i). The fact that $g_b(Z,0)=0$ follows by taking $\rho=0$, which is admissible. The monotonicity of $g_b(Z,\cdot)$ is immediate, since we minimize the same functional over increasing domains. The convexity and continuity of $g_b(Z,\cdot)$ follow by Lemma \ref{lemm:gb}. The fact that $g_b(Z,\alpha)<0$ when $\alpha>0$ follows by arguments similar to the ones used in Remark \ref{rema:infFeps}.

We prove now (ii). The concavity of $g_b(\cdot,\alpha)$ follows from the fact that it is the infimum of a family of affine functions. The monotonicity of $g_b(\cdot,\alpha)$ is immediate, as well as the fact that $g_b(0,\alpha)=0$. Finally the fact that $\lim_{Z \to +\infty} g(Z, \alpha)=-\infty$ follows from the inequality
$$g_b(Z,\alpha)\le T(\rho_0)+bC(\rho_0)-ZU_0(\rho_0)$$
for any admissible $\rho_0$.
\end{proof}

\section{The case of a single nucleus}\label{s:mono}


In this section, we obtain additional properties of the function $g_b(Z,\alpha)$ for particular choices of the kinetic energy $T$ and of the electronic interaction $C$. We note that this function encodes the full $\Gamma$-limit of the family $(G_\eps)_\eps$ whenever there is only one nucleus placed at the origin $X_1=0$. We shall see that this function can be expressed through a function of one real variable $t$ only, namely
\be\label{elle}
L(t)=\inf\left\{T(\rho)+C(\rho)-\int\frac{\rho(x)}{|x|}\,dx\ :\ \rho\ge0,\ \int\rho(x)\,dx\le t\right\}.
\ee
This function then allows us to study the cases when a ionization phenomenon occurs. Indeed, the ionization phenomenon happens if the non-increasing function $L$ above remains constant after a certain threshold, in which case the limit problem \eqref{limpb} may admit solutions with total mass strictly lower than $m$ when $m$ is large enough (see Section \ref{s:examplesnew}). We first derive a general convexity property for $L$.

\begin{prop}\label{Lstrict}
Let $F$ be a convex functional. Then the function
\be\label{Lconvex}
L(t)=\inf\Big\{F(\rho)\ :\ \int\rho\,dx\le t\Big\}
\ee
is convex. If in addition $F$ is strictly convex and for every $t\ge0$ the infimum in \eqref{Lconvex} is achieved on $\rho_t$ with $\int\rho_t\,dx=t$, then $L$ is strictly convex.
\end{prop}

\begin{proof}
Let $t_1,t_2$ be fixed and let $\rho_1,\rho_2$ be such that $\int\rho_1\,dx\le t_1$ and $\int\rho_2\,dx\le t_2$. Then for every $s\in[0,1]$ we have
$$L\big(st_1+(1-s)t_2\big)\le F\big(s\rho_1+(1-s)\rho_2\big)\le sF(\rho_1)+(1-s)F(\rho_2).$$
Taking the infimum on $\rho_1$ and $\rho_2$ gives the convexity of $L$.

\noindent If $F$ is strictly convex and for every $t$ the infimum in \eqref{Lconvex} is achieved on $\rho_t$ with $\int\rho_t\,dx=t$, we obtain the strict convexity of $L$ by taking $\rho_1$ and $\rho_2$ as the optimal solutions corresponding to $t_1$ and $t_2$ respectively.
\end{proof}

\subsection{The case of Thomas-Fermi $T_0$ with electronic correlation $C_0$}\label{ss31}
In this case the function $g_b(Z,\alpha)$ is given by
$$g_b(Z,\alpha)=\inf\Big\{T_0(\rho)+bC_0(\rho)-Z\int\frac{\rho}{|x|}\,dx\ :\ \int\rho\,dx\le\alpha\Big\}.$$
Writing $\rho(x)=\lambda s^3\eta(sx)$ and choosing $s=b^2/Z$ and $\lambda=(Z/b)^3$ we obtain easily
$$g_b(Z,\alpha)=\frac{Z^3}{b}L(\alpha b^3Z^{-3}),$$
where the function $L$ is given by \eqref{elle} with $T,C$ replaced by $T_0,C_0$, namely
\be\label{l1}
L(t)=\inf\left\{\int\Big(\rho^{5/3}+\frac34\rho^{4/3}-\frac{\rho(x)}{|x|}\Big)\,dx\ :\ \rho\ge0,\ \int\rho(x)\,dx\le t\right\}.
\ee

\begin{theo}\label{T0C0}
The minimum problem \eqref{l1} admits a solution for every $t\ge0$. In addition, the function $L(t)$ is strictly decreasing in $t$ and its limit as $t\to+\infty$ is finite. As a consequence, in this case, for every charge $Z>0$ no ionization occurs, and the limit problem \eqref{limpb} admits a minimum in the class of $\rho$ with $\int\rho\,dx=m$.
\end{theo}

\begin{proof}
If $(\rho_n)$ is a minimizing sequence for problem \eqref{l1}, using that $L(t)\geq 0$ we have for every $R>0$
\[\begin{split}
\int\rho_n^{5/3}\,dx&\le\int_{|x|\le R}\frac{\rho_n}{|x|}\,dx+\int_{|x|\ge R}\frac{\rho_n}{|x|}\,dx\\
&\le\|\rho_n\|_{L^{5/3}}\Big(\int_{|x|\le R}\frac{1}{|x|^{5/2}}\,dx\Big)^{2/5}+\frac{t}{R}\\
&=\|\rho_n\|_{L^{5/3}}(8\pi)^{2/5}R^{1/5}+\frac{t}{R},
\end{split}\]
from which we obtain that $(\rho_n)$ is bounded in $L^{5/3}$. By a similar argument we can see that the functional
$$\rho\mapsto\int\frac{\rho}{|x|}\,dx$$
is continuous for the weak $L^{5/3}$ convergence, while the functional
$$\rho\mapsto T_0(\rho)+C_0(\rho)$$
is weakly lower semicontinuous. In addition, the solution $\bar\rho$ of \eqref{l1} is such that $\int\bar\rho\,dx=t$ because otherwise we would have, by the Euler-Lagrange equation,
$$\frac53\bar\rho^{2/3}+\bar\rho^{1/3}=\frac{1}{|x|}\;,$$
which is impossible, because this would imply $\bar\rho(x)\sim1/|x|^3$ as $|x|\to+\infty$ and then $\int\bar\rho\,dx=+\infty$. Thus there exists $C>0$ such that
$$\frac53\bar\rho^{2/3}+\bar\rho^{1/3}=\Big(\frac{1}{|x|}-C\Big)^+\;.$$
In particular, the solution $\bar\rho$ is compactly supported. In addition, since the solution $\\bar\rho$ for $L(t)$ has total variation $t$ for any $t \geq 0$, the function $L(t)$ is strictly decreasing. Finally, since
\[\begin{split}
&\int_{|x|\le R}\frac{\rho}{|x|}\,dx\le\Big(\int_{|x|\le R}\rho^{5/3}dx\Big)^{3/5}\Big(8\pi R^{-1/2}\Big)^{2/5}\\
&\int_{|x|\ge R}\frac{\rho}{|x|}\,dx\le\Big(\int_{|x|\ge R}\rho^{4/3}dx\Big)^{3/4}\Big(4\pi R^{-1}\Big)^{1/4}
\end{split}\]
the limit of $L(t)$ at $t\to+\infty$
$$L(\infty)=\inf\left\{T_0(\rho)+C_0(\rho)-\int\frac{\rho(x)}{|x|}\Big)\,dx\ :\ \rho\ge0\right\}.$$
is finite, and achieved on a unique function $\rho$. By the Euler-Lagrange equation
$$\frac53\rho^{2/3}+\rho^{1/3}=\frac{1}{|x|}$$
we have that the optimal solution $\rho$ behaves as $1/|x|^3$ as $|x|\to\infty$, and then cannot be in $L^1$.
\end{proof}

\begin{rema}
If we replace $\rho^{5/3}$ by $\rho^p$ and $\rho^{4/3}$ by $\rho^q$, with $q<p$, by repeating the arguments above we obtain:
\begin{itemize}
\item[-]the condition to have the infimum $L(\infty)$ finite is $q<3/2<p$ and in this case there exists a (unique) optimal solution $\rho$ which belongs to $L^p\cap L^q$;
\item[-]the optimal solution $\rho$ is in $L^1$ when $q<4/3$.
\end{itemize}
\end{rema}

\subsection{The case of von Weizs\"acker $W_{\alpha,\beta}$ with electronic correlation $C_0$}\label{ss32}
In this case the function $g_b(Z,\alpha)$ is given by
$$g_b(Z,\alpha)=\inf\Big\{W_{\alpha,\beta}(\rho)+bC_0(\rho)-Z\int\frac{\rho}{|x|}\,dx\ :\ \int\rho\,dx\le\alpha\Big\}.$$
Writing as before $\rho(x)=\lambda s^3\eta(sx)$ and choosing $s=b^{2-\beta}Z^{\beta-1}$ and $\lambda=Z^3/b^3$ we obtain easily
$$g_b(Z,\alpha)=Z^{\beta+3}b^{-\beta-1}L(\alpha b^3Z^{-3}),$$
where the function $L$ is given by \eqref{elle} with $T,C$ replaced by $W_{\alpha,\beta},C_0$. In order to see if a ionization phenomenon occurs for some values of the charge $Z$, as in the previous subsection we study the limit problem for $L(t)$ as $t\to+\infty$:
\be\label{Linfty}
L(\infty)=\inf\left\{W_{\alpha,\beta}(\rho)+C_0(\rho)-\int\frac{\rho(x)}{|x|}\,dx\ :\ \rho\ge0\right\}.
\ee

\begin{theo}
Problem $L(\infty)$ above admits a unique solution $\rho$, which is not in $L^1$. As a consequence, in this case, for every charge $Z>0$ no ionization occurs, and the $\Gamma$-limit functional $G$ in \eqref{defG} admits a minimum in the class of $\rho$ with $\int\rho\,dx=m$.
\end{theo}

\begin{proof}
By Lemma \ref{lemm:TU} we deduce that the infimum $L(\infty)$ is finite. Indeed, we have for every $R>0$
$$\int_{|x|\ge R}\frac{\rho}{|x|}\,dx\le\Big(\int_{|x|\ge R}\rho^{4/3}dx\Big)^{3/4}\Big(\frac{4\pi}{R}\Big)^{1/4}\;,$$
while
\[\begin{split}
\int_{|x|\le R}\frac{\rho}{|x|}\,dx&\le c\big(W_{\alpha,\beta}(\rho)\big)^{1/2}\Big(\int_{|x|\le R}\rho\,dx\Big)^{(\beta+1)/6}\\
&\le c\big(W_{\alpha,\beta}(\rho)\big)^{1/2}\bigg(\Big(\int_{|x|\le R}\rho^{4/3}dx\Big)^{3/4}\Big(\frac{4\pi R^3}{3}\Big)^{1/4}\bigg)^{(\beta+1)/6},
\end{split}\]
which easily lead to the finiteness of $L(\infty)$. In addition, minimizing sequences $(\rho_n)$ are such that $W_{\alpha,\beta}(\rho_n)$ and $\int\rho_n^{4/3}dx$ are bounded, and this leads to the existence of a minimizer $\rho$, which is unique because of the strict convexity of the involved functional.

Our goal is to prove that this solution $\rho$ is not in $L^1$, which implies that in this case, for every charge $Z>0$ no ionization occurs, and the $\Gamma$-limit functional $G$ admits a minimum in the class of $\rho$ with $\int\rho\,dx=m$.

In order to do that, it is convenient to write the problem replacing $\rho$ by $u^\gamma$, in the form
\be\label{Linftyu}
\inf\left\{\gamma^\beta\int|\nabla u|^\beta dx+\frac34\int u^{4\gamma/3}dx-\int\frac{u^\gamma}{|x|}\,dx\ :\ u\ge0\right\}.
\ee
Since the functional in \eqref{Linfty} is convex, the solution $\rho$ is radial; hence the solution $u$ of \eqref{Linftyu} is also radial. Writing the problem in polar coordinates we obtain
$$4\pi\inf\left\{\gamma^\beta\int r^2|u'|^\beta dr+\frac34\int r^2u^{4\gamma/3}dr-\int ru^\gamma\,dr\ :\ u\ge0\right\}$$
whose Euler-Lagrange equation is
$$-\beta\gamma^\beta|u'|^{\beta-2}\Big((\beta-1)u''+\frac{2}{r}u'\Big)+\gamma u^{(4\gamma-3)/3}=\frac{\gamma}{r}u^{\gamma-1}.$$
It is convenient now to set $u(r)=v(r^{-p})$, with $p=(3-\beta)/(\beta-1)$. Denoting by $t$ the variable $r^{-p}$, an easy computation gives for $v(t)$ the equation
\be\label{ELv}
A t^{\beta(p+1)/p}|v'(t)|^{\beta-2}v''(t)=v(t)^{\gamma-1}\Big(v(t)^{\gamma/3}-t^{1/p}\Big),
\ee
with
$$A=\frac{\beta(\beta-1)p^\beta}{\gamma}>0.$$
Then
$$\begin{cases}
v(t)\text{ is convex over the curve }v=t^{3/(\gamma p)}\\
v(t)\text{ is concave under the curve }v=t^{3/(\gamma p)}.
\end{cases}$$
If near $t=0$ we have $v(t)\ge c>0$, then it is immediate to see that $u(r)\ge c$ near infinity and so $\rho(r)$ cannot be in $L^1$. So let us consider the case
$$v(0)=\lim_{t\to0^+}v(t)=0.$$
According to the values of the parameter $\beta$, for the curve $v=v=t^{3/(\gamma p)}$ we have the situations illustrated in Figure \ref{fig1}.

\begin{figure}[h!]
\centering
{\includegraphics[scale=0.53]{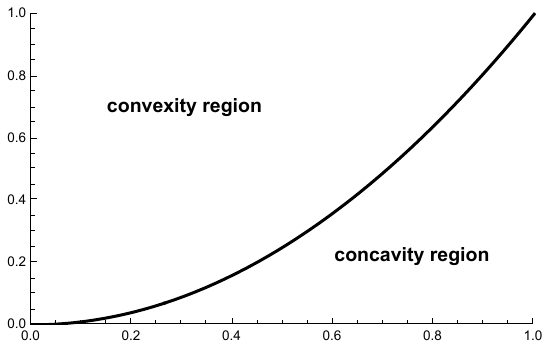}
\includegraphics[scale=0.53]{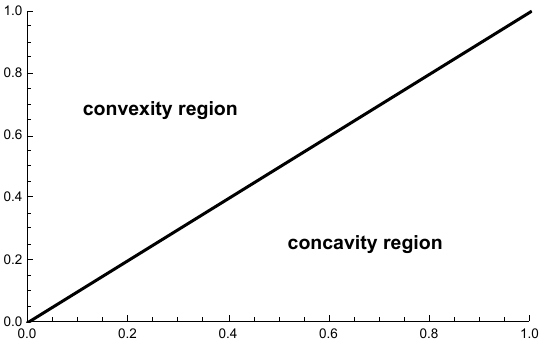}
\includegraphics[scale=0.53]{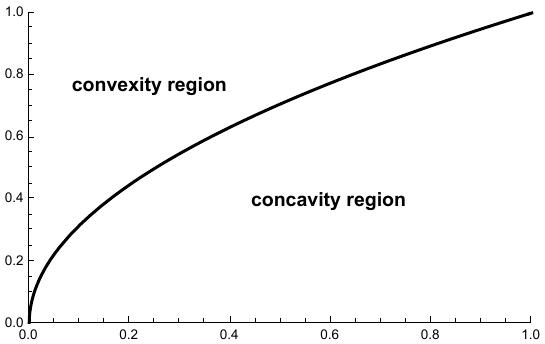}}
\caption{The case $\beta>5/3$ (left), $\beta=5/3$ (middle), $\beta<5/3$ (right).}\label{fig1}
\end{figure}

If $\beta>5/3$, the curve $v=t^{3/(\gamma p)}$ is convex. Since we are supposing $v(0)=0$, and since $v(t)$ is convex over the curve $v=t^{3/(\gamma p)}$, in a neighborhood of the origin we must have $v(t)\ge t^{3/(\gamma p)}$. Then $u(r)\ge r^{-3/\gamma}$ near infinity, which gives $\rho(r)\ge r^{-3}$. Hence $\rho$ is not in $L^1$.

\medskip

If $\beta=5/3$, the curve $v=t^{3/(\gamma p)}$ is linear, and in a neighborhood of the origin we must have $v(t)\sim kt$ for a suitable $k>0$. Then $u(r)\sim kr^{-2}$ near infinity, which gives $\rho(r)\sim k^{3/2}r^{-3}$. Hence $\rho$ is not in $L^1$.

\medskip

We consider now the case $\beta<5/3$ in which the curve $v=t^{3/(\gamma p)}$ is concave. Then $v(t)$ must have at least a linear growth near the origin, being in the concavity region of Figure \ref{fig1}. We cannot have $v(t)^{\gamma/3}\ll t^{1/p}$; indeed, from \eqref{ELv} this would imply
$$\big|(v')^{\beta-2}v''\big|\sim v^{\gamma-1}t^{(1-\beta p-\beta)/p}\ge t^k$$
near the origin, with
$$k=\gamma-1+\frac1p-\beta-\frac\beta p.$$
This implies
$$v'(t)\ge t^{(k+1)/(\beta-1)}$$
and
$$v(t)\ge t^{(k+\beta)/(\beta-1)}.$$
This is impossible, since we assumed $v(t)\ll t^{3/(\gamma p)}$, while 
$$\frac{3}{\gamma p}>\frac{k+\beta}{\beta-1}\qquad\text{when }\beta\in[5/4,5/3[,$$
as it can be easily verified. Then $v(t)\sim t^{3/(\gamma p)}$ near the origin, which gives $u(r)\sim r^{-3/\gamma}$ near infinity, so that $\rho(r)\sim r^{-3}$, which is not in $L^1$.
\end{proof}

\subsection{The case of Thomas-Fermi $T_0$ with electronic correlation $D(\rho)$}\label{ss33}
This case has been considered in \cite{LS}; we include it here for the sake of completeness. The function $g_b(Z,\alpha)$ is then given by
$$g_b(Z,\alpha)=\inf\Big\{T_0(\rho)+bD(\rho)-Z\int\frac{\rho}{|x|}\,dx\ :\ \int\rho\,dx\le\alpha\Big\}.$$
Writing as before $\rho(x)=\lambda s^3\eta(sx)$ and choosing $s=b^{2/3}Z^{1/3}$ and $\lambda=Z/b$ we obtain easily that
$$g_b(Z,\alpha)=Z^{7/3}b^{-1/3}L(\alpha b/Z),$$
where the function $L$ is given by \eqref{elle} with $T,C$ replaced by $T_0,D$. We summarize here below some properties of the function $L(t)$ in the case of the present subsection. A more extensive study can be found in \cite{LS}.

\begin{prop}\label{L(t)}
For the function
$$L(t)=\inf\Big\{T_0(\rho)+D(\rho)-\int\frac{\rho}{|x|}\,dx\ :\ \int\rho\,dx\le t\Big\}$$
we have:
\begin{itemize}
\item[(i)]the minimizer $\rho$ for $L(1)$ is supported on the entire $\R^3$, and $\rho(x)\approx(3/\pi)^3|x|^{-6}$ as $|x|\to\infty$;
\item[(ii)]$L$ is differentiable, and $L'(1)=0$;
\item[(iii)]for every $t<1$ the minimizer $\rho_t$ for $L(t)$ has a compact support, and $\spt\rho_t\subset B(0,R_t)$ with $R_t\le1/|L'(t)|$;
\item[(iv)]$L(t)\approx-3(\pi/2)^{4/3}t^{1/3}$ as $t\to0$; in particular $L'(0^+)=-\infty$.
\end{itemize}
\end{prop}

\begin{proof}
The proof of (i) and (iv) can be found in \cite{LS}, Theorem IV.10 and Theorem II.32 respectively.

The proof of (ii) is also contained in Theorem II.32 of \cite{LS}. We give here a simple proof of $L'(1)=0$. If $\rho$ is the minimizer for $L(1)$, by taking $t\rho$ as a test function for $L(t)$, we have
$$L(1)\le L(t)\le t^{5/3}T_0(\rho)+t^2D(\rho)-t\int\frac{\rho}{|x|}\,dx.$$
Denoting by $h(t)$ the right-hand side in the inequality above, we have that $h(t)$ is differentiable, with minimum at $t=1$, and $h(1)=L(1)$. This proves that $L'(1)=0$.

Let us prove (iii). Let $t<1$ and let $\rho_t$ be the minimizer for $L(t)$. The function $\rho_t$ is radial and we denote by $R_t\le+\infty$ the radius of its support. For every $R<R_t$ we set
$$\rho_R=\rho_t1_{B_R},\qquad t_R=\int\rho_R\,dx<t.$$
Then
\[\begin{split}
L(t_R)&\le T_0(\rho_R)+D(\rho_R)-\int\frac{\rho_R}{|x|}\,dx\\
&\le T_0(\rho_t)+D(\rho_t)-\int\frac{\rho_t}{|x|}\,dx+\int_{\{|x|\ge R\}}\frac{\rho_t}{|x|}\,dx\\
&\le L(t)+\frac1R\int_{\{|x|\ge R\}}\rho_t\,dx\\
&=L(t)+\frac1R(t-t_R).
\end{split}\]
This implies that $|L'(t)|\le1/R_t$, so that whenever $L'(t)\ne0$ we get that $R_t$ is finite and $R_t\le1/|L'(t)|$. The fact that $L'(t)\ne0$ for $t<1$ follows from Theorem \ref{ioni1} below.
\end{proof}

In order to see if a ionization phenomenon occurs for some values of the charge $Z$, as in the previous subsections we study the limit problem for $L(t)$ as $t\to+\infty$:
\be\label{Linfty3}
L(\infty)=\inf\left\{\int\rho^{5/3}dx+D(\rho)-\int\frac{\rho(x)}{|x|}\,dx\ :\ \rho\ge0\right\}.
\ee

\begin{theo}\label{ioni1}
The infimum $L(\infty)$ above is finite, and this problem admits a unique solution $\rho$, which moreover satisfies $\int\rho\,dx=1$. As a consequence, in this case, when $m\le Z/b$ no ionization occurs, and the $\Gamma$-limit functional $G$ in \eqref{defG} admits a minimum in the class of $\rho$ with $\int\rho\,dx=m$. On the contrary, if $m>Z/b$, the ionization phenomenon occurs and the minimum of the $\Gamma$-limit functional $G$ in \eqref{defG} is attained in the class of $\rho$ with $\int\rho\,dx=\alpha$, with $\alpha=Z/b<m$.
\end{theo}

\begin{proof}
By H\"older inequality we have for every $R>0$ and $\eps>0$
\be\label{431}
\begin{split}
\int_{|x|\le R}\frac{\rho}{|x|}\,dx&\le C\Big(\int_{|x|\le R}\rho^{5/3}dx\Big)^{3/5}R^{1/5}\\
&\le\frac{\eps}{2}\int_{|x|\le R}\rho^{5/3}dx+\frac{C_\eps}{2}R^{1/2}
\end{split}
\ee
for a suitable constant $C_\eps$. On the other hand we infer from Lemma \ref{lemm:DU} that
\be\label{432}
\int_{|x|\ge R}\frac{\rho(x)}{|x|}\,dx\le\Big(\frac4R D(\rho)\Big)^{1/2}\le\frac{\eps}{2}D(\rho)+\frac{C_\eps}{2}\frac{1}{R}.
\ee
Putting together \eqref{431} and \eqref{432} we obtain
$$\int\frac{\rho(x)}{|x|}\,dx\le \eps\Big(\int\rho^{5/3}dx+D(\rho)\Big)+C_\eps\Big(R^{1/2}+\frac1R\Big),$$
which implies that $L(\infty)$ is finite. In addition, minimizing sequences $(\rho_n)$ are such that $\int\rho_n^{5/3}dx$ and $D(\rho_n)$ are bounded, and this leads to the existence of a minimizer $\bar\rho$, which is unique because of the strict convexity of the involved functional.

Our goal is to prove that this solution $\bar\rho$ is in $L^1$, which implies that in this case, for every charge $Z\ge mb\big(\|\bar\rho\|_{L^1}\big)^{-1/3}$ no ionization occurs, and the $\Gamma$-limit functional $G$ admits a minimum in the class of $\rho$ with $\int\rho\,dx=m$. On the contrary, below the threshold above, a fraction of electrons escapes to infinity (ionization phenomenon) and the $\Gamma$-limit functional $G$ in \eqref{defG} admits a minimum with total mass strictly less than $m$.

The Euler-Lagrange equation for the minimizer in \eqref{Linfty3} is
\be\label{ELTD}
\frac53\bar\rho^{2/3}+w=\frac{1}{|x|}
\ee
where $w$ is the solution of the PDE
$$-\Delta w=4\pi\bar\rho.$$
Taking the Laplacian of both sides in \eqref{ELTD} gives
\be\label{normu}
-\frac53\Delta(\bar\rho^{2/3})+4\pi\bar\rho=4\pi\delta_0.
\ee
We now multiply both sides of \eqref{normu} by a smooth function $\phi_R$ such that
$$\phi_R(x)=\begin{cases}
1&\text{if }|x|\le R\\
0&\text{if }|x|\ge2R
\end{cases}\qquad\text{and}\qquad0\le\phi_R(x)\le1$$
and integrate. We obtain
$$-\frac53\int_{R\le|x|\le2R}\bar\rho^{2/3}\Delta\phi_R\,dx+4\pi\int\bar\rho\phi_R\,dx=4\pi$$
which gives, as $R\to\infty$, $\|\bar\rho\|_{L^1}=1$, as soon as we find $\phi_R$ such that
$$\int_{R\le|x|\le2R}\bar\rho^{2/3}|\Delta\phi_R|\,dx\to0.$$
Since $\bar\rho$ is the minimizer of problem \eqref{Linfty3}, we have $\bar\rho\in L^{5/3}$ so that, by H\"older inequality, it is enough to find $\phi_R$ such that
\be\label{zero}
\int_{R\le|x|\le2R}|\Delta\phi_R|^{5/3}\,dx\to0.
\ee
In polar coordinates the integral in \eqref{zero} is
$$4\pi\int_R^{2R}r^2\Big|\phi_R''(r)+\frac2r\phi_r'(r)\Big|^{5/3}dr$$
which becomes, with $\phi_R(r)=\psi(R/r)$,
$$4\pi\int_R^{2R}r^2\Big|\psi''(R/r)\frac{R^2}{r^4}\Big|^{5/3}dr=4\pi R^{-1/3}\int_{1/2}^1 s^{8/3}|\psi''(s)|^{5/3}ds.$$
Taking for instance
$$\psi(s)=(2s-1)^2(3-2s^2)$$
we obtain that the right-hand side above vanishes as $R\to\infty$. Thus, in this case the ionization threshold is $m=Z/b$.
\end{proof}

\begin{rema}
Replacing the Thomas-Fermi term $\frac35\int\rho^{5/3}dx$ by $\frac1p\int\rho^pdx$ with $p>3/2$, we may repeat all the arguments in the proof of Theorem \ref{ioni1} and conclude that also in this more general case $L(\infty)$ is finite and there exists a unique solution $\rho$ of \eqref{Linfty3} which is in $L^1$. Moreover, we have $\int\rho\,dx=1$, so that for every $p>3/2$ the ionization threshold remains $m=Z/b$.
\end{rema}

\subsection{The case of von Weizs\"acker $W_{\alpha,\beta}$ with electronic correlation $D(\rho)$}
In this case the function $g_b(Z,\alpha)$ is given by
$$g_b(Z,\alpha)=\inf\Big\{W_{\alpha,\beta}(\rho)+bD(\rho)-Z\int\frac{\rho}{|x|}\,dx\ :\ \int\rho\,dx\le\alpha\Big\}.$$
Writing as before $\rho(x)=\lambda s^3\eta(sx)$ and choosing $s=b^{(2-\beta)/3}Z^{(\beta+1)/3}$ and $\lambda=Z/b$ we obtain easily
$$g_b(Z,\alpha)=Z^{(7+\beta)/3}b^{-(\beta+1)/3}L(\alpha b/Z),$$
where the function $L$ is given by \eqref{elle} with $T,C$ replaced by $W_{\alpha,\beta},D$. Again, the ionization phenomenon occurs when the problem
$$L(\infty)=\inf\left\{W_{\alpha,\beta}(\rho)+ D(\rho)-\int\frac{\rho(x)}{|x|}\,dx\ :\ \rho\ge0\right\}$$
admits a solution $\bar\rho$ which is in $L^1$, and in this case the ionization threshold is $m=\|\bar\rho\|_{L^1}Z/b$.

The following lemmas will be useful.

\begin{lemm}\label{deltaR}
Let $\rho\ge0$, with $\rho\in L^1_{loc}$, be such that $\int\rho\,dx>1$. Then there exist $R>0$ and $\delta>0$ such that
$$\rho*\frac{1}{|x|}\ge\frac{1+\delta}{|x|}\qquad\text{whenever }|x|>R.$$
\end{lemm}

\begin{proof}
We have for every $t>0$
\[\begin{split}
\rho*\frac{1}{|x|}&=\int\frac{\rho(y)}{|x-y|}\,dy\ge\int_{\{|y|\le t\}}\frac{\rho(y)}{|x-y|}\,dy\\
&\ge\int_{\{|y|\le t\}}\frac{\rho(y)}{|x|+|y|}\,dy\ge\frac{1}{|x|+t}\int_{\{|y|\le t\}}\rho\,dy.
\end{split}\]
Setting
$$m(t)=\int_{\{|y|\le t\}}\!\!\!\rho\,dy$$
we have to find $R>0$ and $\delta>0$ such that
$$\frac{m(t)}{r+t}\ge\frac{1+\delta}{r}\qquad\text{whenever }r>R$$
or equivalently
$$m(t)\ge(1+\delta)\Big(1+\frac{t}{R}\Big).$$
By the assumption $\int\rho\,dx>1$ we can choose $t$ such that $m(t)>1$ and then take
$$\delta=\frac{m(t)-1}{2},\qquad R=t\,\frac{m(t)+1}{m(t)-1}$$
which concludes the proof.
\end{proof}

\begin{lemm}\label{rhoL1}
Let $u\ge0$ be a radial function such that
$$\begin{cases}
\ds-\Delta_\beta u+\frac{c}{|x|}u^{\gamma-1}\le0\quad\text{on }\{|x|>R\}\\
\ds\int\frac{u^\gamma(x)}{|x|}\,dx<+\infty,
\end{cases}$$
with $1<\beta\le2$, $\gamma=3\beta/(5-\beta)$, $c>0$, $R>0$. Then
$$\int u^\gamma\,dx<+\infty.$$
\end{lemm}

\begin{proof}
In polar coordinates we have
$$-\beta|u'|^{\beta-2}\Big((\beta-1)u''+\frac2r u'\Big)+\frac{c}{r}u^{\gamma-1}\le0.$$
Now we set $u(r)=v(r^{-p})$ with $p=(3-\beta)/(\beta-1)$, and we use the variable $t=r^{-p}$ with $T=R^{-p}$. We obtain
$$-\beta(\beta-1)p^\beta t^{\beta(p+1)/p}|v'(t)|^{\beta-2}v''(t)+ct^{1/p}v^{\gamma-1}(t)\le0\quad\text{on }\{]0,T[\}.$$
Since $u$ is supposed nonnegative, we have that the function $v(t)$ is convex, so the limit $v(0)=\lim_{t\to0^+}v(t)$ exists. By the assumption
$$\int_{\{|x|>R\}}\frac{u^\gamma(x)}{|x|}\,dx=\frac{4\pi}{p}\int_0^T t^{-(2+p)/p}v^\gamma(t)\,dt<+\infty$$
we deduce that $v(0)=0$ and, by convexity, $v'(t)\ge0$. Then, multiplying by $v'$ we obtain
$$\beta(\beta-1)p^\beta t^{(\beta p+\beta-1)/p}\Big(\frac{(v')^\beta}{\beta}\Big)'\ge c\Big(\frac{v^\gamma}{\gamma}\Big)'.$$
Integrating on $[0,t]$ we have
\be\label{ineq}
\begin{split}
c\frac{v^\gamma}{\gamma}&\le\beta(\beta-1)p^\beta\bigg[-\int_0^t\frac{\beta p+\beta-1}{p}s^{(\beta-1)(p+1)/p}\frac{(v')^\beta}{\beta}\,ds+t^{(\beta p+\beta-1)/p}\frac{(v')^\beta}{\beta}\bigg]\\
&\le(\beta-1)p^\beta t^{(\beta p+\beta-1)/p}(v')^\beta.
\end{split}\ee
Since $v(t)$ is convex with $v(0)=0$ we have that $v'(t)$ is bounded near the origin. Then
\be\label{bound}
v^\gamma\le Ct^{(\beta p+\beta-1)/p}
\ee
near the origin, for a suitable constant $C>0$. Then, in order to show that $u^\gamma$ is integrable, thanks to the assumption $\int\frac{u^\gamma(x)}{|x|}\,dx<+\infty$ it is enough to show that $u^\gamma$ is integrable near infinity. By \eqref{bound} we have
\[\begin{split}
\int_{|x|>M}u^\gamma\,dx&=4\pi\int_M^\infty r^2 u^\gamma(r)\,dr=4\pi\int_0^{M^{-p}}v^\gamma(t)t^{-(3+p)/p}dt\\
&\le4\pi C\int_0^{M^{-p}}t^{(\beta p+\beta-1)/p} t^{-(3+p)/p}dt,
\end{split}\]
which is integrable when $\beta<2$ since
$$\frac{\beta p+\beta-1}{p}-\frac{3+p}{p}>-1.$$
The case $\beta=2$ needs a different proof. From \eqref{ineq} we have
$$\frac{c}{2}v^2\le t^3(v')^2$$
from which
$$\sqrt{c/2}\,t^{-3/2}\le\frac{v'}{v}.$$
Integrating on $[t,t_0]$ we have
$$-\sqrt{2c}\big(t_0^{-1/2}-t^{-1/2}\big)\le\log v(t_0)-\log(v(t).$$
For $t$ close to the origin this gives
$$\log v(t)\le-\sqrt{2c}t^{-1/2}.$$
Hence
$$v(t)\le\exp\big(-\sqrt{2c}t^{-1/2}\big),$$
which implies the integrability of $u^\gamma$.
\end{proof}

We are now in a position to consider the limit problem for the function $L(t)$ as $t\to\infty$:
\be\label{WD}
\min\left\{W_{\alpha,\beta}(\rho)+ D(\rho)-\int\frac{\rho(x)}{|x|}\,dx\ :\ \rho\ge0\right\}.
\ee

\begin{theo}
The minimization problem \eqref{WD} admits a unique solution $\rho$ which is in $L^1$.
\end{theo}

\begin{proof}
We first prove that the infimum in \eqref{WD} is finite. Let $\rho \in L^1(\R^3)$ for which the energy functional in \eqref{WD} is finite, then from Lemma \ref{lemm:TU} we get that
\[
\forall R>0, \qquad 
\int_{|x|<R} \frac{1}{|x|}\rho(x)\,dx\le
\kappa\,\big(W_{\alpha,\beta}(\rho)\big)^{1/2}\Big(\int_{|x| < R}\rho\,dx\Big)^{(\beta+1)/6}
\]
and since for any $R>0$ one has
\[
\int_{|x| < R}\rho\,dx \;\le\; R \, \int_{|x|<R} \frac{1}{|x|}\rho(x)
\]
we conclude
\[
\forall R>0, \qquad 
\int_{|x|<R} \frac{1}{|x|}\rho(x)\,dx\le \kappa\,R^{(\beta+1)/(5-\beta)}\big(W_{\alpha,\beta}(\rho)\big)^{3/(5-\beta)}
\]
for some constant $\kappa$. Thanks to Lemma \ref{lemm:DU}, we then infer
\[
\int \frac{1}{|x|}\rho(x)\,dx\le \kappa\,R^{(\beta+1)/(5-\beta)}\big(W_{\alpha,\beta}(\rho)\big)^{3/(5-\beta)} + \left( \frac{R}{2} D(\rho) \right)^{1/2}\,. 
\]
Fixing $R>0$ small enough so that $\kappa\,R^{(\beta+1)/(5-\beta)}<1$, and noting that $3/(5-\beta)\le1$, we conclude that the infimum in $\eqref{WD}$ is finite, and that the problem has at least one minimizer $\rho$.

It is now convenient to set $\rho=u^\gamma$ with $\gamma=3\beta/(5-\beta)$; in this way the functional to be minimized becomes
$$\gamma^\beta\int|\nabla u|^\beta\,dx+ D(u^\gamma)-\int\frac{u^\gamma}{|x|}\,dx,$$
which gives the Euler-Lagrange equation
$$-\beta\gamma^\beta\Delta_\beta u+ \gamma u^{\gamma-1}\Big(u^\gamma*\frac{1}{|x|}\Big)-\gamma\frac{u^{\gamma-1}}{|x|}=0.$$
If $\int\rho\,dx=\int u^\gamma\,dx\le1$ we have the required summability of the solution $\rho$; otherwise, by Lemma \ref{deltaR} we obtain that for suitable $R>0$ and $\delta>0$
$$\rho*\frac{1}{|x|}=u^\gamma*\frac{1}{|x|}\ge\frac{1+\delta}{|x|}\qquad\text{whenever }|x|>R.$$
Then we have
$$-\beta\gamma^\beta\Delta_\beta u+\gamma\delta\frac{u^{\gamma-1}}{|x|}\le0\qquad\text{whenever }|x|>R.$$
Now Lemma \ref{rhoL1} applies, which provides the summability of $u^\gamma$, hence of $\rho$.
\end{proof}

\begin{rema}\label{rem:BBL}
In \cite{BBL}, a similar existence result in $L^1$ was obtained for the problem:
$$\inf\left\{W_{-1,2}(\rho)+\frac1{p}\int\rho^pdx+D(\rho)-ZU_0(\rho)\ :\ \rho\in L^1(\R^3),\ \rho\ge0\right\}.$$
where the additional term $\frac1p\int\rho^pdx$ is added (with $1<p<+\infty$). 
\end{rema}

\begin{rema}
It would be interesting to consider the case $C(\rho)=C_{GC}(\rho)$, with either the kinetic energy $T_0$ or $W_{\alpha,\beta}$, and to establish if in this case the ionization phenomenon occurs. Again, this amounts to consider the minimization problems
\[\begin{split}
&\inf\left\{\int\rho^{5/3}dx+C_{GC}(\rho)-\int\frac{\rho(x)}{|x|}\,dx\ :\ \rho\ge0\right\}\\
&\inf\left\{W_{\alpha,\beta}(\rho)+C_{GC}(\rho)-\int\frac{\rho(x)}{|x|}\,dx\ :\ \rho\ge0\right\}
\end{split}\]
and to see if their solutions are in $L^1$ or not. Some discussions on this problem are made in \cite{dmlene22}
when the kinetic energy is dropped while the Coulomb potential is replaced by a continuous function.
\end{rema}

\section{Some examples}\label{s:examplesnew}

In this section we present two examples that show how the electronic charge distributes when several nuclei are present. We consider the case of two nuclei and we take into account the energy $G$ in \eqref{defG} obtained as the $\Gamma$-limit of energies as $\eps\to0$. We also take, for simplicity, the parameter $b=1$. More precisely, we consider the case of two nuclei, with charge $Z_1$ and $Z_2$ respectively, and of a total electronic charge $m$ that has to be distributed around the two nuclei according to the minimum problem
\be\label{pbgb}
\min\big\{g_1(Z_1,m_1)+g_1(Z_2,m_2)\ :\ m_1+m_2\le m\big\}.
\ee
When no ionization occurs the inequality in \eqref{pbgb} is saturated, and $m_1+m_2=m$; on the contrary, in case of ionization the inequality can be strict.

\begin{exam}\label{exam1}
In this first example we consider the Thomas-Fermi kinetic energy
$$T_0(\rho)=\int\rho^{5/3}dx$$
together with the electronic repulsion
$$C_0(\rho)=\frac34\int\rho^{4/3}dx.$$
As we have seen in Subsection \ref{ss31} in this case no ionization occurs, and the function $g_1$ is given by
$$g_1(Z,m)=Z^3L(mZ^{-3}).$$
The problem is then reduced to the minimization
$$\min\Big\{Z_1^3L(m_1Z_1^{-3})+Z_2^3L(m_2Z_2^{-3})\ :\ m_1+m_2=m\Big\}.$$
By Proposition \ref{Lstrict}, in this case the function $L$ is strictly convex, hence
\[\begin{split}
Z_1^3 L\Big(\frac{m_1}{Z_1^3}\Big)+Z_2^3 L\Big(\frac{m_2}{Z_2^3}\Big)&\ge(Z_1^3+Z_2^3)L\Big(\frac{Z_1^3}{Z_1^3+Z_2^3}\frac{m_1}{Z_1^3}+\frac{Z_2^3}{Z_1^3+Z_2^3}\frac{m_2}{Z_2^3}\Big)\\
&=(Z_1^3+Z_2^3)L\Big(\frac{m}{Z_1^3+Z_2^3}\Big)\\
&=Z_1^3 L\Big(\frac{m}{Z_1^3+Z_2^3}\Big)+Z_2^3 L\Big(\frac{m}{Z_1^3+Z_2^3}\Big).
\end{split}\]
Then the optimal distribution of electrons is
$$m_1=\frac{mZ_1^3}{Z_1^3+Z_2^3},\qquad m_2=\frac{mZ_2^3}{Z_1^3+Z_2^3}.$$
\end{exam}

\begin{exam}\label{exam2}
In this second example we still consider the Thomas-Fermi kinetic energy $T_0$ as above, but with the electronic repulsion
$$D(\rho)=\int\int\frac{\rho(x)\rho(y)}{|x-y|}dx\,dy.$$
We still take for simplicity $b=1$. We have seen in Subsection \ref{ss33} that in this case the ionization phenomenon occurs for $Z<m$, which amounts to say that the corresponding function $L(t)$ is constant for $t\ge1$. Since in this case we have
$$g_1(Z,m)=Z^{7/3}L(m/Z),$$
the problem is reduced to the minimization
$$\min\Big\{Z_1^{7/3}L(m_1/Z_1)+Z_2^{7/3}L(m_2/Z_2):\ m_1+m_2\le m,\ m_1\le Z_1,\ m_2\le Z_2\Big\}.$$
Clearly, if $m\ge Z_1+Z_2$, the optimal distribution of electrons is
$$m_1=Z_1,\qquad m_2=Z_2.$$
We consider now the case $m<Z_1+Z_2$. By differentiating, we obtain
$$Z_1^{4/3}L'(m_1/Z_1)=Z_2^{4/3}L'(m_2/Z_2)$$
and the values of $m_1,m_2$ can be then obtained through a numerical approximation of the function $L$. By the properties seen in Proposition \ref{L(t)}, we have that
$$0<m_1<Z_1,\qquad0<m_2<Z_2;$$
in addition, by the strict convexity of the function $L$, if $Z_1<Z_2$ we deduce that
$$\frac{m_1}{Z_1}<\frac{m_2}{Z_2}.$$
\end{exam}

\section{Proof of Theorem \ref{th:mainG}}\label{s:proofG}

We now turn to the proof of Theorem \ref{th:mainG}, which follows directly from the Lemmata \ref{lemm:gliminf} and \ref{lemm:glimsup} below. We first state some usefull properties for the function $g_b$.

\begin{lemm}\label{lemm:gb}
For fixed $b,Z\ge0$, the function $\alpha\mapsto g_b(Z,\alpha)$ is convex and continuous on $\R_+$. Moreover, one has
\begin{multline}\label{e:gbcompact}
g_b(Z,\alpha)=\inf\bigg\{T(\rhotild)+b\,C(\rhotild)-Z\int\frac{\rhotild(x)}{|x|}\,dx\ :\\
\rhotild \in \M_+,\ supp(\rhotild) \mbox{ compact}, \ \rhotild(\R^3) = \alpha\bigg\}.
\end{multline}
\end{lemm}

\begin{proof}
The convexity of $\alpha\mapsto g_b(Z,\alpha)$ follows from the convexity of the functional $T+b\,C-Z\,U_0$. We note that $g_b(Z,0)=0$, so that $g_b$ is non-positive on $\R_+$. From \eqref{h:TU} we infer that for all $\rhotild \in \M_+$ for which $T(\rhotild)<+\infty$ it holds 
\[
T(\rhotild)+b \ C(\rhotild)-Z\,\int\frac{\rhotild(x)}{|x|}\,dx \;\geq\;
T(\rhotild)+b \ C(\rhotild)-Z\,K_U\,(T(\rhotild+1)^q \rhotild(\R^3)^p
\]
for some $q\in]0,1]$, from wich we infer that $g_b(Z,\alpha)$ takes finite values. As a consequence it is continuous on $]0,+\infty[$. The continuity at $0^+$ follows from the lower semi-continuity of the functional $T+b\,C-Z\,U_0$.

We now turn to \eqref{e:gbcompact}, and first prove the following
\begin{multline}\label{e:gbcompactint}
g_b(Z,\alpha)=\inf\bigg\{T(\rhotild)+b\,C(\rhotild)-Z\int\frac{\rhotild(x)}{|x|}\,dx\ :\\
\rhotild\in\M_+,\ supp(\rhotild)\mbox{ compact},\ \rhotild(\R^3)\le\alpha\bigg\}
\end{multline}
Remind that $g_b(Z,\alpha)$ is finite, and fix $\rhotild\in\M_+$ almost optimal for $g_b(Z,\alpha)$. Consider a sequence $(\theta_n)_n$ of cut-off functions as in \eqref{h:Ttrunc}, so that $T(\theta_n \rhotild) \to T(\rhotild)$ as $n \to +\infty$. Then from the monotonicity \eqref{h:Cmono} of $C$ we have that
$$\forall n,\qquad C(\theta_n\rhotild)\le\C(\rhotild)$$
and by the lower-semicontinuity of $C$ we also get $C(\theta_n \rhotild) \to C(\rhotild)$ as $n \to +\infty$. Finally we note that
$$\forall n\ge1,\qquad\norm{\int\frac{\theta_n(x)\rhotild(x)}{|x|}\,dx-\int\frac{\rhotild(x)}{|x|}\,dx}
\le\frac{\alpha}{n}$$
which concludes the proof of \eqref{e:gbcompactint}.

Now assume that $\rhotild$ has compact support and is almost optimal in \eqref{e:gbcompactint}. Let $\rhotild^*$ be smooth, compactly supported with total mass $\alpha - \rhotild(\R^3)$. We consider the family in $\M_+$ given by
\[
\rhotild_{\eps,X} = [\rhotild^*(\cdot-X)]_{\#\eps}
\]
where $\eps>0$ and $X \in \R^3$.
Then we have $[\rhotild+\rhotild_{\eps,X}](\R^3) = \alpha$, and we only need to prove that for adequate choices of the parameters $\eps$ and $X$ it holds
\[
(T+b\,C-Z\,U_0)(\rhotild+\rhotild_{\eps,X}) \simeq (T+b\,C-Z\,U_0)(\rhotild)
\]
We infer from the scaling properties \eqref{h:Thomog} and \eqref{h:Chomog}
\[
T(\rhotild_{\eps,X})\to0\quad\mbox{and}\quad C(\rhotild_{\eps,X})\to0\quad\mbox{as}\quad\eps\to0.
\]
Now we consider $\norm{X}$ big enough so that the supports of $\rhotild$ and $\rhotild_{\eps,X}$ do not intersect. In this case we infer from \eqref{h:Tadd} that
$T(\rhotild+\rhotild_{\eps,X}) \simeq T(\rhotild)$.
From \eqref{h:Cmono} and \eqref{h:Csubadd} we get that
\[
C(\rhotild)\le C(\rhotild+\rhotild_{\eps,X})\le C(\rhotild)+C(\rhotild_{\eps,X}) + K(R) \rhotild(\R^3)(1-\rhotild(\R^3))
\]
where $R$ is the distance between the supports of $\rhotild$ and $\rhotild_{\eps,X}$.
Since $R \to +\infty$ and $K(R) \to 0$ as $\norm{X} \to +\infty$, we infer
that $C(\rhotild+\rhotild_{\eps,X}) \simeq C(\rhotild)$ for $\norm{X}$ big enough.

Finally we notice that
\[
\norm{U_0(\rhotild_{\eps,X})}\le\frac{1-\rhotild(\R^3)}{dist(0,support(\rhotild_{\eps,X}))}\to0\quad\mbox{as}\quad\norm{X}\to+\infty.
\]
which concludes the proof since $U_0$ is linear.
\end{proof}

\begin{lemm}[$\Gamma-\liminf$ inequality]\label{lemm:gliminf}
Under the hypotheses of Theorem \ref{th:mainG}, for every $\rho\in\M_+$
it holds
$$\Gamma-\liminf_{\eps \to 0} G_\eps (\rho) \geq G(\rho) := \sum_{i=1}^N g_b(Z_i,\rho(\{X_i\}))\,.$$
\end{lemm}

\begin{proof}
Let $\rho \in \M_+$
, and denote by $(\rho_\eps)_\eps$ a family weakly* converging in $\M_+$ to $\rho$, we aim to show
$$\liminf_{\eps \to 0} G_\eps (\rho_\eps)\ge G(\rho)\,.$$
Without loss of generality we may assume that the left hand side is finite, and up to considering a subsequence, we can also suppose that the $\liminf$ is a in fact a limit, so that there exists a constant $K_G$ such that
$$\forall\eps>0,\qquad G_\eps(\rho_\eps)\le K_G<+\infty\,.$$
We also note that $(\rho_\eps(\R^3))_\eps$ is bounded as $\eps \to 0$ since it converges weakly to $\rho$.

\emph{Step 1.} We first prove that the sequence $(\eps^2 T(\rho_\eps))_\eps$ is bounded as $\eps \to 0$. To this end, we use the property \eqref{h:TU} to compute
\begin{align*}
\eps U(\rho_\eps)
&=\sum_{i=1}^M \eps\int\frac{Z_i\rho_\eps}{\norm{\,\cdot-X_i}}\\
&=\sum_{i=1}^M\eps\int\frac{Z_i\rho_\eps(\cdot+X_i)}{\norm{\,\cdot\,}}\\
&=\sum_{i=1}^M\int\frac{Z_i[\rho_\eps(\cdot+X_i)]_{\#\eps}}{\norm{\,\cdot\,}}\\
&\le M K_U(T([\rho_\eps(\cdot+X_i)]_{\#\eps})+1)^q\,\rho_\eps(\R^3)^p \sum_i Z_i
\end{align*}
for some constants $K_U$, $q\in \,]0,1[\,$ and $p \in \,]0,1]$.
From the translation invariance of $T$ and the fact that $(\rho_\eps(\R^3))_\eps$ is bounded as $\eps \to 0$ we get that
$$\eps U(\rho_\eps)\le\kappa\,(T([\rho_\eps]_{\# \eps})+1)^q$$
for some constant $\kappa$. Using the scaling property \eqref{h:Thomog} for $T$ and the fact that $C$ is non negative we get
$$T([\rho_\eps]_{\# \eps})-\kappa(T([\rho_\eps]_{\# \eps})+1)^q\le\eps^2 T(\rho_\eps)-\eps U(\rho_\eps)\le G_\eps(\rho_\eps) \leq K_G\,.$$
Since $q<1$ this implies that $T([\rho_\eps]_{\# \eps})=\eps^2 T(\rho_\eps)$ is bounded as $\eps \to 0$.

\emph{Step 2.} We now localize the sequence $(\rho_\eps)_\eps$ around the poles $X_i$. To this end, let $\theta_\delta$ be a cut-off function verifying \eqref{h:Tlocal}, with $0<\delta < \frac{1}{5} \min\{\norm{X_i-X_j} : i \neq j\}$, so that
\[
T(\theta_\delta\,\rho_\eps)\le T(\rho_\eps) + K_T (T(\rho_\eps) + 1)^r \omega(\delta)
\]
with $r \in [0,1]$.
From the monotonicity \eqref{h:Cmono} of $C$ we get
\[
C(\theta_\delta \rho_\eps)\le C(\rho_\eps)\,.
\]
From the fact that $(\rho_\eps(\R^3))_\eps$ is bounded and $\norm{\,\cdot-X_i} > \delta$ on the support of $1-\theta_\delta$ we have
\[
-U(\theta_\delta\rho_\eps)=-U(\rho_\eps)+U((1-\theta_\delta)\rho_\eps)\le-U(\rho_\eps)-\frac{\kappa}{\delta}
\]
for some constant $\kappa$. As a consequence, since $\eps^2 T(\rho_\eps)$ is bounded we obtain
\[
\liminf_{\eps\to0}G_\eps(\theta_\delta\rho_\eps)\le\kappa\omega(\delta)+\liminf_{\eps\to0}G_\eps(\rho_\eps)
\]
for some constant $\kappa$ independant of $\delta$.

\emph{Step 3.} We now prove 
\be\label{e:Gdeltaeps}
\liminf_{\eps \to 0} G_\eps (\theta_\delta \rho_\eps) \geq \sum_{i=1}^M g_b(Z_i,\rho(B(X_i,2\delta))
\ee
To this end, we define 
\[
\theta_{\delta,i} := 1_{B(X_i,2\delta)} \theta_\delta
\]
and note that $\theta_\delta = \sum \theta_{\delta,i}$.
Since the supports of the functions $\theta_{\delta_i}$ are at distance at least $\delta$, we can use the additivity \eqref{h:Tadd} of $T$ and the superadditivity \eqref{h:Csuperadd} of $C$ respectively, as well as their invariance by translation, to get
\[
T(\theta_\delta \rho_\eps) = \sum_{i=1}^M T(\theta_{\delta,i} \rho_\eps)
=\sum_{i=1}^M T([\theta_{\delta,i} \rho_\eps](\cdot+X_i))
\]
and
\[
C(\theta_\delta \rho_\eps) \geq \sum_{i=1}^M C(\theta_{\delta,i} \rho_\eps) - \kappa K_C(\delta) = \sum_{i=1}^M C([\theta_{\delta,i} \rho_\eps](\cdot+X_i)) - \kappa K_C(\delta)
\]
where $\kappa$ is a uniform bound for $\rho_\eps(\R^3)\rho(\R^3)$.
On the other hand we have, since $(\rho_\eps(\R^3))$ is bounded we also have
\begin{align*}
- U(\theta_\delta \rho_\eps)
& = - \sum_{i,j=1}^M \int \frac{Z_j \theta_{\delta,i}\rho_\eps}{\norm{\,\cdot-X_j}} \\
& = - \sum_{i,j=1}^M \int \frac{Z_j [\theta_{\delta,i}\rho_\eps](\cdot+X_i)}{\norm{\,\cdot+X_i-X_j}}
\geq -\sum_{i=1}^M U_0([\theta_{\delta,i}\rho_\eps](\cdot+X_i)) - \frac{\kappa}{\delta}
\end{align*}
where $\kappa$ is a constant.

From the respective scaling properties of $T$, $C$ and $U_0$, we get that for all $i,\eps$ it holds
\[\begin{split}
&\eps^2 T([\theta_{\delta,i} \rho_\eps](\cdot+X_i)) + \eps C([\theta_{\delta,i} \rho_\eps](\cdot+X_i)) -\eps U_0([\theta_{\delta,i} \rho_\eps](\cdot+X_i))\\
&\qquad=T([[\theta_{\delta,i} \rho_\eps](\cdot+X_i)]_{\sharp \eps} )
+C([[\theta_{\delta,i} \rho_\eps](\cdot+X_i)]_{\sharp \eps} )
-U_0([[\theta_{\delta,i} \rho_\eps](\cdot+X_i)]_{\sharp \eps})\\
&\qquad\ge g_b(Z_i,[[\theta_{\delta,i}\rho_\eps](\cdot+X_i)]_{\sharp\eps}(\R^3))\\
&\qquad\ge g_b([[\theta_{\delta,i} \rho_\eps](\cdot+X_i)]_{\sharp \eps}(\R^3),Z_i)\ge g_b(\rho_\eps(B(X_i,2\delta)),Z_i),
\end{split}\]
where the last line holds since the map $\alpha\mapsto g_b(Z,\alpha)$ is non-increasing. Finally, since $\alpha\mapsto g_b(Z,\alpha)$ is continuous, we infer from the previous estimates that \eqref{e:Gdeltaeps} holds.

From Steps 2 and 3 we obtain that
$$\liminf_{\eps\to0}G_\eps (\rho_\eps)\ge\sum_i g_b(Z_i,\rho(B(X_i,2\delta)))-\kappa\,\omega(\delta)$$
and we conclude the proof by using again the fact that $\alpha \mapsto g_b(Z,\alpha)$ is continuous and $\omega(\delta)\to 0$ as $\delta\to0$.
\end{proof}

\begin{lemm}[$\Gamma-\limsup$ inequality]\label{lemm:glimsup}
Under the hypotheses of Theorem \ref{th:mainG}, for every $\rho \in \M_+$
it holds
\be\label{e:glimsup}
\Gamma-\limsup_{\eps\to0}G_\eps(\rho)\le G(\rho):=\sum_{i=1}^N g_b(Z_i,\rho(\{X_i\}))\,.
\ee
\end{lemm}

\begin{proof}
Let $\rho\in\M_+$. 

\emph{Step 1.} We prove \eqref{e:glimsup} in the special case where $\rho$ is of the form
\be\label{e:rhoreg}
\rho = \sum_{i=1}^M \rho(\{X_i\}) \delta_{X_i} + \rhotild
\ee
where $\rhotild$ is in $C^\infty_c$ with support at positive distance from $\{X_1,\ldots,X_M\}$. We fix $\delta>0$, then from \eqref{e:gbcompact} we infer that for
every $i$ there exists $\rhotild_i \in \M_+$ with compact support such that
\[
g_b(Z_i,\rho(\{X_i\}))\ge T(\rhotild_i)+b \ C(\rhotild_i)-Z_i\,\int\frac{\rhotild_i(x)}{|x|}\,dx - \delta
\]
and $\rhotild_i(\R^3) = \rho(\{X_i\})$.
We now set
\[
\rho_\eps := \sum_{i=1}^M \eps^{-3}\theta_k\Big(\frac{\cdot-X_k}{\eps}\Big)+\rhotild = \sum_{i=1}^M [\rhotild_i(\cdot-X_i)]_{\# \eps^{-1}}+\rhotild\, 
\]
Note that $\rho_\eps$ is such that $\int d\rho_\eps =\int d\rho$.
For $\eps >0$ small enough, the supports of the functions in the previous sum are mutually disjoint, so we may assume that they are at distance $\Delta>0$ from each other. Then the additivity of $T$, as well as its scaling property \eqref{h:Thomog}, give
\[
\eps^2 T(\rho_\eps) = \eps^2 \sum_{i=1}^M T([\rhotild_i(\cdot-X_i)]_{\# \eps^{-1}})+ \eps^2 T(\rhotild)
= \sum_{i=1}^M T(\rhotild_i)+ \eps^2 T(\rhotild)\,.
\]
On the other hand, we use the subadditivity \eqref{h:Csubadd} of $C$ and its invariance by translation to get
\begin{align*}
\eps C(\rho_\eps)&\le\eps\sum_{i=1}^M C([\rhotild_i(\cdot-X_i)]_{\# \eps^{-1}})+\eps C(\rhotild)+\eps\kappa\\
&=\sum_{i=1}^M C(\rhotild_i)+\eps C(\rhotild)+\eps\kappa
\end{align*}
where the constant $\kappa$ depends only on the constant $K_C(\Delta)$ as well as the masses of the $\rhotild_i$ and $\rhotild$. Finally we have that
\begin{align*}
-\eps U(\rho_\eps)&=-\eps\sum_{i,j=1}^M Z_j\int\frac{\eps^{-3}\rhotild_i((x-X_i)/\eps)}{\norm{x -X_j}}\,dx-\eps U(\rhotild)\\
&\le-\sum_{i=1}^M Z_i\int\frac{\rhotild_i(x)}{\norm{x}}\,dx+\eps\frac{\kappa}{\Delta}
\end{align*}
where the constant $\kappa$ depends only on the masses of the $\rhotild_i$ and $\rhotild$.
As a consequence of the preceding estimates we infer that for $\eps >0$ small enough we have
\[
G_\eps(\rho_\eps)\le\sum_{i=1}^M g_b(Z_i,\rho(\{X_i\}))+M\delta+\eps C(\rhotild)+\eps\kappa+\eps\frac{\kappa}{\Delta}\,.
\]
Passing to the limit as $\eps\to0$ gives
\[
\limsup_{\eps\to0}G_\eps(\rho_\eps)\le\sum_{i=1}^M g_b(Z_i,\rho(\{X_i\}))+M \delta
\]
and we obtain \eqref{e:glimsup} in this case since $\delta>0$ is arbitrary.

\emph{Step 2.} Let now $\rho\in\M_+$. Then there exists a sequence $\rho_n$ of the form \eqref{e:rhoreg} such that
\[
\rho_n\weak\rho\qquad\mbox{and}\qquad\forall i,n,\quad\rho_n(\{X_i\})=\rho(\{X_i\}).
\]
with $\int d\rho_n = \int d\rho$ for all $n$.
Hence $G(\rho_n)=G(\rho)$. Therefore, the inequality \eqref{e:glimsup} can be extended to $\rho$ by the lower semicontinuity of the functional $\Gamma-\limsup G_\eps$.
\end{proof}

%
%
%
%
%

%
%
%
%
%

\bigskip\ack
This paper has been written during some visits of the authors at the Departments of Mathematics of Universities of Firenze and of Pisa, and at the Laboratoire IMATH of University of Toulon. The authors gratefully acknowledge the warm hospitality and support of these institutions. In particular, the second author warmly thanks the French CNRS and the FRUMAM of Marseille, which through the {\it``postes rouges''} program allowed him to spend a research stay in Toulon.

The work of the second author is part of the PRIN project {\it Gradient flows, Optimal Transport and Metric Measure Structures} (2017TEXA3H), funded by the Italian Ministry of Education and Research (MIUR). The work of the fourth author is partially financed by the {\it``Fondi di ricerca di ateneo''} of the University of Firenze and partially financed by the EU-Next Generation EU, (Missione 4, Componente 2, Investimento 1.1 {\it Progetti di Ricerca di Rilevante Interesse Nazionale} (PRIN), CUPB53D23009310006 - (2022J4FYNJ). The second and fourth authors are members of the research group GNAMPA of INdAM.


\bigskip
{\small\noindent
Guy Bouchitt\'e:
Laboratoire IMATH, Universit\'e de Toulon\\
BP 20132, 83957 La Garde Cedex - FRANCE\\
{\tt bouchitte@univ-tln.fr}

\bigskip\noindent
Giuseppe Buttazzo:
Dipartimento di Matematica, Universit\`a di Pisa\\
Largo B. Pontecorvo 5, 56127 Pisa - ITALY\\
{\tt giuseppe.buttazzo@unipi.it}\\
{\tt http://www.dm.unipi.it/pages/buttazzo/}

\bigskip\noindent
Thierry Champion:
Laboratoire IMATH, Universit\'e de Toulon\\
BP 20132, 83957 La Garde Cedex - FRANCE\\
{\tt champion@univ-tln.fr}\\
{\tt http://champion.univ-tln.fr}

\bigskip\noindent
Luigi De Pascale:
Dipartimento di Matematica e Informatica, Universit\`a di Firenze\\
Viale Morgagni 67/a, 50134 Firenze - ITALY\\
{\tt luigi.depascale@unifi.it}\\
{\tt http://web.math.unifi.it/users/depascal/}

\end{document}